\documentclass[apa,notitlepage,twocolumn,nofootinbib,superscriptaddress]{revtex4-1}
\usepackage{mathtools}
\usepackage{amsmath}
\usepackage[shortlabels]{enumitem}
\usepackage{thmtools}
\usepackage{thm-restate}
\usepackage{etoolbox}
\usepackage{amsmath,amsfonts,amssymb,array,mathtools,multirow,bm,tcolorbox,relsize,booktabs}
\usepackage[utf8]{inputenc}
\usepackage[T1]{fontenc}
\usepackage{graphicx,epic,times,eepic,epsfig,latexsym,verbatim,color}
\usepackage{amsfonts}       
\usepackage{nicefrac}       
\usepackage[linesnumbered,ruled,procnumbered]{algorithm2e}
\usepackage{bbm}
\usepackage{float}
\usepackage{tikz}
\usetikzlibrary{chains}
\usetikzlibrary{fit}
\usetikzlibrary{positioning}
\usetikzlibrary{shapes.geometric}
\usetikzlibrary{calc}
\definecolor{tensorblue}{rgb}{0.8,0.9,1}
\tikzset{ten/.style={fill=tensorblue}}
\newcommand{\diagram}[1]{ \begin{array}{cc}\begin{tikzpicture}[scale=.5,every node/.style={sloped,allow upside down},baseline={([yshift=+0ex]current bounding box.center)}] #1 \end{tikzpicture} \end{array} }

\usepackage{epsfig}
\usetikzlibrary{shapes.symbols,patterns} 
\usepackage{pgfplots}

\usepackage[strict]{changepage}
\usepackage{hyperref}
\hypersetup{colorlinks=true,citecolor=blue,linkcolor=blue,filecolor=blue,urlcolor=blue,breaklinks=true}

\usepackage[marginal]{footmisc}
\usepackage{url}
\usepackage{theorem}

\newtheorem{definition}{Definition}
\newtheorem{proposition}{Proposition}
\newtheorem{lemma}[proposition]{Lemma}
\newtheorem{fact}{Fact}
\newtheorem{theorem}[proposition]{Theorem}
\newtheorem{corollary}[proposition]{Corollary}

\def\squareforqed{\hbox{\rlap{$\sqcap$}$\sqcup$}}
\def\qed{\ifmmode\squareforqed\else{\unskip\nobreak\hfil
\penalty50\hskip1em\null\nobreak\hfil\squareforqed
\parfillskip=0pt\finalhyphendemerits=0\endgraf}\fi}
\def\endenv{\ifmmode\;\else{\unskip\nobreak\hfil
\penalty50\hskip1em\null\nobreak\hfil\;
\parfillskip=0pt\finalhyphendemerits=0\endgraf}\fi}
\newenvironment{proof}{\noindent \textbf{{Proof~} }}{\hfill $\blacksquare$}

\newcounter{remark}

\newcounter{example}

\mathchardef\ordinarycolon\mathcode`\:
\mathcode`\:=\string"8000
\def\vcentcolon{\mathrel{\mathop\ordinarycolon}}
\begingroup \catcode`\:=\active
  \lowercase{\endgroup
  \let :\vcentcolon
  }

\usepackage{cleveref}
\usepackage{graphicx}

\newcommand{\nc}{\newcommand}
\nc{\rnc}{\renewcommand}
\nc{\lbar}[1]{\overline{#1}}
\nc{\bra}[1]{\langle#1|}
\nc{\ket}[1]{|#1\rangle}
\nc{\dketbra}[2]{\vert #1 \rangle \hspace{-.8mm} \rangle \hspace{-.4mm} \langle\hspace{-.8mm}\langle #2 \vert}
\nc{\dbra}[1]{\langle\hspace{-.8mm}\langle #1\vert}
\nc{\dket}[1]{\vert#1\rangle\hspace{-.8mm}\rangle}
\nc{\ketbra}[2]{|#1\rangle\!\langle#2|}
\nc{\braket}[2]{\langle#1|#2\rangle}
\newcommand{\braandket}[3]{\langle #1|#2|#3\rangle}

\nc{\proj}[1]{| #1\rangle\!\langle #1 |}
\nc{\avg}[1]{\langle#1\rangle}
\nc{\rank}{\operatorname{Rank}}
\nc{\smfrac}[2]{\mbox{$\frac{#1}{#2}$}}
\nc{\tr}{\operatorname{Tr}}
\nc{\ox}{\otimes}
\nc{\dg}{\dagger}
\nc{\dn}{\downarrow}
\nc{\cA}{{\cal A}}
\nc{\cB}{{\cal B}}
\nc{\cC}{{\cal C}}
\nc{\cD}{{\cal D}}
\nc{\cE}{{\cal E}}
\nc{\cF}{{\cal F}}
\nc{\cG}{{\cal G}}
\nc{\cH}{{\cal H}}
\nc{\cI}{{\cal I}}
\nc{\cJ}{{\cal J}}
\nc{\cK}{{\cal K}}
\nc{\cL}{{\cal L}}
\nc{\cM}{{\cal M}}
\nc{\cN}{{\cal N}}
\nc{\cO}{{\cal O}}
\nc{\cP}{{\cal P}}
\nc{\cQ}{{\cal Q}}
\nc{\cR}{{\cal R}}
\nc{\cS}{{\cal S}}
\nc{\cT}{{\cal T}}
\nc{\cU}{{\cal U}}
\nc{\cV}{{\cal V}}
\nc{\cX}{{\cal X}}
\nc{\cY}{{\cal Y}}
\nc{\cZ}{{\cal Z}}
\nc{\cW}{{\cal W}}
\nc{\csupp}{{\operatorname{csupp}}}
\nc{\qsupp}{{\operatorname{qsupp}}}
\nc{\var}{{\operatorname{var}}}
\nc{\rar}{\rightarrow}
\nc{\lrar}{\longrightarrow}
\nc{\polylog}{{\operatorname{polylog}}}
\nc{\wt}{{\operatorname{wt}}}
\nc{\av}[1]{{\left\langle {#1} \right\rangle}}
\nc{\supp}{{\operatorname{supp}}}
\nc{\VComb}{{\widetilde{\cal C}}}
\nc{\VChoi}{{\widetilde{C}}}

\nc{\RR}{{{\mathbb R}}}
\nc{\CC}{{{\mathbb C}}}
\nc{\FF}{{{\mathbb F}}}
\nc{\NN}{{{\mathbb N}}}
\nc{\ZZ}{{{\mathbb Z}}}
\nc{\PP}{{{\mathbb P}}}
\nc{\QQ}{{{\mathbb Q}}}
\nc{\UU}{{{\mathbb U}}}
\nc{\EE}{{{\mathbb E}}}
\nc{\id}{{\operatorname{id}}}

\newcommand{\CPTP}{\text{\rm CPTP}}

\definecolor{colortwo}{rgb}{0.4,0.77,0.17}
\definecolor{colorthree}{rgb}{0.01,0.51,0.93}

\allowdisplaybreaks

\begin{document}
\title{Reversing Unknown Quantum Processes via Virtual Combs\\ for Channels with Limited Information}

\author{Chengkai Zhu}
\thanks{Chengkai Zhu and Yin Mo contributed equally to this work.}
\author{Yin Mo}
\thanks{Chengkai Zhu and Yin Mo contributed equally to this work.}
\author{Yu-Ao Chen}
\author{Xin Wang}
\email{felixxinwang@hkust-gz.edu.cn}
\affiliation{Thrust of Artificial Intelligence, Information Hub,\\
The Hong Kong University of Science and Technology (Guangzhou), Guangzhou 511453, China}

\begin{abstract}
The inherent irreversibility of quantum dynamics for open systems poses a significant barrier to the inversion of unknown quantum processes. To tackle this challenge, we propose the framework of virtual combs that exploit the unknown process iteratively with additional classical post-processing to simulate the process inverse. Notably, we demonstrate that an $n$-slot virtual comb can exactly reverse a depolarizing channel with one unknown noise parameter out of $n+1$ potential candidates, and a 1-slot virtual comb can exactly reverse an arbitrary pair of quantum channels. We further explore the approximate inversion of an unknown channel within a given channel set. A worst-case error decay of $\cO(n^{-1})$ is unveiled for depolarizing channels within a specified noise region. Moreover, we show that virtual combs can universally reverse unitary operations and investigate the trade-off between the slot number and the sampling overhead.
\end{abstract}

\date{\today}
\maketitle

\textcolor{black}{\textbf{\emph{Introduction.---}}}
Suppose a physical apparatus is provided that is guaranteed to perform some unknown process $\cN$, it can be regarded as a black box with no more prior information. Is it possible to simulate the inverse of this process by employing this black box multiple times? 
For such a task of executing a desired transformation based on the given operations, the most comprehensive method entails using a quantum network~\cite{Chiribella_2008a}. Formally, the problem here is
to construct a feasible quantum network $\cC$, connected to the black box $\cN$ for $n$ times, to perform its inverse satisfying $\cC(\cN^{\otimes n})\circ\cN$ be the identity channel, where such a quantum network is generally an $n$-slot quantum comb~\cite{Chiribella_2008a,Chiribella2008a}. The significance of this task lies in revealing fundamental capabilities and properties of quantum operations~\cite{Caruso2014,Holevo2013a}, insights to quantum algorithm design~\cite{Gilyen2018,Martyn2021,Wang2022}, and applications to quantum error cancellation~\cite{Cai2022}. Understanding the power of quantum channels can shed further light on theoretical and applied quantum physics~\cite{Wilde2017book,Hayashi2017b}.

A simple strategy is to apply process tomography to obtain the full matrix representation which is usually resourceful~\cite{Chuang_1997,haah2023queryoptimal}. 
If the unknown process is restricted to unitary operations, numerous works have been carried out to explore efficient methods that can implement the inverse of any unknown unitary (see, e.g.,~\cite{Trillo_2020,Navascu_s_2018,Trillo_2023,Schiansky_2023,Chiribella_2016,sardharwalla2016universal,Sedl_k_2019,Quintino_2019a,Quintino2019}). Recently, deterministic and exact protocols for reversing any unknown unitary have been discovered for qubit case~\cite{Yoshida2023} and arbitrary dimensions~\cite{chen2024quantum}, indicating full knowledge of the process through tomography is not necessary for this task.
Nevertheless, how to extend such protocols to cases where the process is a general quantum channel remains an open question. 

The challenge of reversing general unknown quantum processes is twofold. First, the inverse map of a quantum channel is generally not a physical process as it is not even positive. Such unphysical inverse maps fall under a broader scope of quantum operations, specifically Hermitian-preserving and trace-preserving linear maps. 
Second, even though we know that all such linear maps are simulatable via sampling quantum operations and post-processing~\cite{Buscemi2013,Jiang2020} or measurement-controlled post-processing~\cite{zhao2023power}, implementing the inverse map via existing methods unavoidably requires the complete description of the quantum process.

In this paper, to explore the full potential of reversing an unknown quantum process, we introduce the notion of \textit{virtual combs} by lifting the positivity requirement on quantum combs.
Physically, a virtual comb corresponds to sampling quantum combs with positive and negative coefficients and performing post-processing.
We find an affirmative answer that simulating the inverse of an unknown channel can be achieved with a virtual comb under certain conditions. Taking into account the unknown channel belonging to a given set without any prior information about its specific identity, we find that for arbitrary given two quantum channels, the exact inverse could always be realized with a 1-slot virtual comb without knowing which specific channel is provided. For depolarizing channels, we unveil the remarkable capability of an $n$-slot virtual comb to exactly reverse a depolarizing channel with an unknown noise parameter among $n+1$ possible candidates (see Fig.~\ref{fig:main_schem}). Intriguingly, we also establish a no-go theorem, elucidating the impossibility of a virtual comb to universally reverse an arbitrary quantum channel with finite uses of the channel.

Beyond exact inversion, our investigation extends to approximately reversing unknown quantum channels through virtual combs. For depolarizing channels within an arbitrary noise region, we find a protocol with worst-case error decay of $\cO(n^{-1})$ using $n$ calls of the channel. Notably, it shows the potential application of virtual combs in error cancellation, where our protocol works for mitigating depolarizing noises without requiring prior knowledge of noise parameters. Furthermore, virtual combs are applied to reverse unknown unitary operations. We show that a 1-slot virtual comb suffices to reverse any $d$-dimensional unitary operation and explore its relationship with the previous unitary inversion problem.
Our findings offer fresh perspectives on the interplay between information reversibility and irreversibility in quantum dynamics and provide new avenues for higher-order quantum transformations.

\textcolor{black}{\textbf{\emph{Exact channel inversion.---}}}
Since the inverse of a quantum channel is not necessarily completely positive and a legitimate quantum comb must adhere to be completely positive~\cite{choi1975completely}, to explore the inversion task and fully explore the power of supermaps, we introduce the \textit{virtual comb} as follows.
\begin{definition}[Virtual comb]
Let $\cC_0,\cC_1,\cdots,\cC_{l-1}$ be quantum combs. An affine combination of them $\VComb = \sum_{i=0}^{l-1}\eta_i\cC_i$ is called a virtual comb where $\sum_{i=0}^{l-1}\eta_i = 1, \eta_i\in\mathbb{R}, \forall i$.
\end{definition}
We remark that the name \textit{virtual comb} is bestowed for two reasons: first, its functionality extends beyond that of a conventional comb; second, its virtual nature considering negative values of $\eta_i$ is manifested through its feasibility, achieved by sampling its quasiprobability decomposition and subsequent post-processing.

Now we first focus on a scenario where the quantum process is guaranteed to be within a set of depolarizing channels characterized by varying degrees of noise. We show that the exact inversion of all channels in this set can be achieved with explicit construction of the virtual comb. Furthermore, we establish a no-go theorem for this task, which highlights the limitations in process reversibility imposed by quantum mechanics.

Specifically, the unknown channel belongs to a family of $d$-dimensional depolarizing channels with $m$ elements $\{\cD_{p_1}, \cD_{p_2}, ... \cD_{p_m}\}$, where $\cD_{p}(\cdot) = (1-p)(\cdot) + p I_d/d$ is a depolarizing channel with a noise parameter $p$. The task is to implement the inverse of an arbitrary channel $\cD_{p_i}$ by querying the unknown channel $n$ times. Based on this setting, we present our main result as follows.

\begin{theorem}[Depolarizing channel inversion]\label{thm:main_Ndepo}
  For any $n\geq 1$, let $\cD_{p_1},..., \cD_{p_{n+1}}$ be $n+1$ $d$-dimensional depolarizing channels with distinct noise parameters $p_1,...,p_{n+1} \in[0,1)$. There exists an $n$-slot virtual comb $\VComb$ satisfying
  \begin{equation}\label{Eq:depo_inv_vcomb}
    \VComb(\cD_{p_i}^{\ox n}) = \cD_{p_i}^{-1}, \quad \forall \, i=1,...,n+1.
  \end{equation}
\end{theorem}

The main idea is to utilize the symmetry condition, whereby $\VComb(\cD_{p_i}^{\ox n})$ can be decomposed into a combination of the identity channel and the depolarizing channel. Based on this, we can formulate Eq.~\eqref{Eq:depo_inv_vcomb} into a linear system, and derive a solvability condition and the corresponding construction for the virtual comb. Detailed proofs of the theorems in this manuscript are deferred to appendix. Theorem~\ref{thm:main_Ndepo} unveils an intrinsic application of the virtual comb framework, enabling the exact inversion of a family of depolarizing channels. Remarkably, the protocol applies to a set of noises, and the number of distinct channels within the set that it can exactly reverse increases with the number of slots.

To highlight the unique power of reversing a family of depolarizing channels with unknown noises provided by virtual combs, we note that such an exact channel inversion task cannot be accomplished via a quantum comb, even probabilistically. We defer the detailed statement and proof in appendix.


\begin{figure}[t]
   \centering
   \includegraphics[width=1\linewidth]{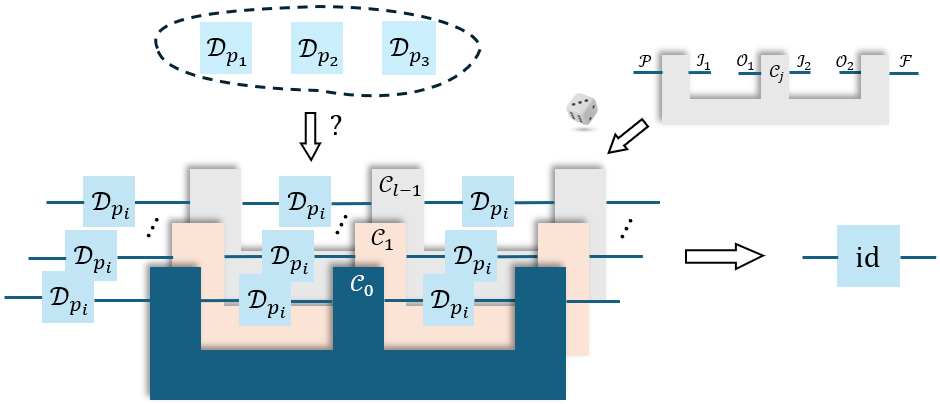}
   \caption{Schematic diagram of reversing depolarizing channels with unknown parameters via a 2-slot virtual comb. A virtual comb is represented as a quasi-probabilistic mixture of quantum combs $\VComb=\sum_{j=0}^{l-1}\eta_j \cC_j$ where $\cC_j$ is a quantum comb. The systems of a quantum comb $\cC_j$ are labeled as $\cP,\cI_i,\cO_i,\cF$. Given a depolarizing channel $\cD_{p_i}$ with an unknown parameter $p_i$ out of three distinct choices, $\VComb$ can exactly reverse $\cD_{p_i}$ by $\VComb(\cD_{p_i}^{\ox 2})\circ \cD_{p_i} = \id$ for $i=1,2,3$.}
   \label{fig:main_schem}
\end{figure}

Significantly, we also obtain a no-go theorem that no $n$-slot virtual comb can be universally capable of exactly reversing every set of $n+2$ channels. This is indicated by the fact that the theoretical maximum for an $n$-slot virtual comb to reverse a collection of depolarizing channels exactly is limited to $n+1$.

\begin{theorem}\label{thm:main_Ndepo_nogo}
  For any $n \geq 1$, let $\cD_{p_1},..., \cD_{p_{n+2}}$ be  $n+2$ $d$-dimensional depolarizing channels with distinct noise parameters $p_1,...,p_{n+2} \in[0,1)$. There is no $n$-slot virtual comb $\VComb$ such that $\VComb(\cD_{p_i}^{\ox n}) = \cD_{p_i}^{-1}, \forall \, i=1,...,n+2$.
\end{theorem}
Theorem~\ref{thm:main_Ndepo_nogo} exposes the inherent limit in reversing an unknown channel, affirming that virtual combs with finite slots cannot achieve the inversion of arbitrary unknown quantum channels.

\textcolor{black}{\textbf{\emph{Approximate channel inversion.---}}}
Although Theorem~\ref{thm:main_Ndepo_nogo} imposes restrictions on achieving the exact inversion for arbitrary quantum channels with a determined virtual comb, the approximate inversion is not prohibited. In approximate inversion, given a set of quantum channels $\Theta=\{(p_i, \cN_i)\}_i$ where $p_i$ is the prior probability for $\cN_i$, we want to find an $n$-slot virtual comb that can make $\VComb(\cN_i^{\ox n}) \circ \cN_i$ as close to the identity channel `$\id$' as possible. With an $n$-slot virtual comb $\VComb$, the average error for reversing the channel set $\Theta$ can be expressed as 
\begin{equation*}\label{err_ave_finite}
e_{\rm ave}^n(\VComb, \Theta) = \frac{1}{2}\sum_{i=1}^m p_i \left\| \VComb(\cN_i^{\ox n}) \circ \cN_i - \id \right\|_\diamond,
\end{equation*}
where $\|\cF\|_{\diamond}:= \sup_{k\in\mathbb{N}}\sup_{\|X\|_1\leq 1}\|(\cF\ox\id_k)(X)\|_1$ denotes the diamond norm of a linear operator $\cF$. The worst-case error is defined as
\begin{equation*}\label{err_worst_finite}
    e_{\rm wc}^n(\VComb, \Theta) = \max \left\{\frac{1}{2}\left\| \VComb(\cN_i^{\ox n}) \circ \cN_i - \id \right\|_\diamond: \cN_i \in \Theta\right\}.
\end{equation*}
Note that for any two HPTP maps $\cN_1,\cN_2$ from system $A$ to $B$, the completely bounded trace distance can be evaluated by semidefinite programming (SDP), which is a powerful tool in quantum information~\cite{watrous2009semidefinite,XWthesis,Skrzypczyk2023}. Then the optimal average error for approximately reversing quantum channels within the set $\Theta$ is determined via an SDP, the details of which are provided in appendix.

\begin{figure}[t]
\centering
    \centering
    \includegraphics[width=.95\linewidth]{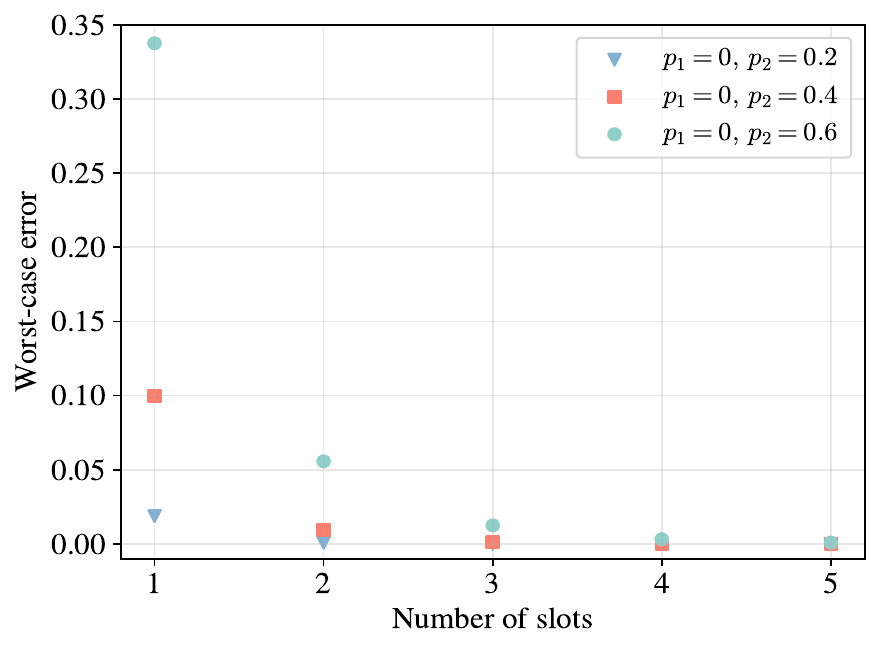}
    \caption{Upper bounds on the minimum worst-case error for reversing a depolarizing channel with an unknown noise parameter $p\in[p_1, p_2]$. The $x$-axis represents the number of slots in a virtual comb. 
    }
    \label{fig:wc_err_depo}
\end{figure}

Based on the numerical calculations of the SDP, we present intriguing results for reversing general quantum channels. The numerical calculations are implemented in MATLAB~\cite{MATLAB} with the interpreters CVX~\cite{cvx,gb08} and QETLAB~\cite{qetlab}. In each experiment, we generate $m$ random qubit-to-qubit quantum channels by the proposed measures in~\cite{Kukulski_2021}, e.g., generating random Choi operators, and calculate the average error of approximately reversing them (with equal prior probability) via 1-slot virtual combs by SDP. Then we apply the \textit{computer-assisted proofs} given in Ref.~\cite{Bavaresco_2021} to construct a feasible solution for the virtual combs. It is worth noting that when $m\leq 13$, we observe that the average errors across $1000$ experimental iterations remain consistently below a tolerance of $1\cdot 10^{-5}$ whenever the channels are invertible. However, intriguingly, when $m\geq 14$, the average errors increase up to the first decimal place, indicating that no virtual comb could achieve near-exact inversion for all these channels. Therefore, we conjecture an upper limit of $13$ elements in the channel set for reversing quantum channels using virtual combs. Notably, we present a theorem demonstrating that the exact inversion is always achievable for any pair of quantum channels.

\begin{theorem}[General channel inversion]\label{thm:two_general_channels}
  For any two invertible quantum channels $\cN_1$, $\cN_2$, there exists a $1$-slot virtual comb $\VComb$ satisfying $\VComb(\cN_i) = \cN_i^{-1}, \forall \, i=1,2$.
\end{theorem}

Theorem~\ref{thm:two_general_channels} highlights the remarkable capability of a 1-slot virtual comb to reverse an \textit{arbitrary} given pair of quantum channels, even when the input and output systems of the channels have different dimensions. 

For depolarizing channels, we now consider that the noise levels are not a few fixed values but fall within a specified range $[p_1, p_2]$.
As an $n$-slot virtual comb $\VComb$ could exactly reverse $n+1$ distinct noise level, using the unknown channel more times is surely to enhance performance. Here we show that as the number of slots in the virtual comb (or the calls for the channel) increases, the worst-case error in channel inversion diminishes at least at a rate of $\cO(n^{-1})$.

\begin{theorem}\label{thm:wc_err}
    Let $0\leq p_1<p_2\leq 1$, the minimum worst-case error of approximately reversing a depolarizing channel $\cD_{p}$ with $p\in [p_1, p_2]$ using an $n$-slot virtual comb is at most $\cO(n^{-1})$.
\end{theorem}

This result can be simply understood as follows: by Theorem~\ref{thm:main_Ndepo}, we can exactly reverse a depolarizing channel whose noise parameter is from $\big\{p_1+(p_2-p_1)k/n\big\}_{k=0}^{n}$ via an $n$-slot virtual comb. Then, in a continuous case, we demonstrate that the worst-case error is at most $\cO(n^{-1})$ within each interval. Based on this scheme, we present the upper bounds of the minimum worst-case error for the cases $p_1, p_2$ are $(0,0.2)$, $(0,0.4)$, and $(0,0.6)$ in Fig.~\ref{fig:wc_err_depo}. A detailed explanation can be found in appendix.
As the number of calls to the unknown channel increases, performance rapidly converges to exact for all these channels.

\textcolor{black}{\textbf{\emph{Application to error cancellation of unknown depolarizing noises.---}}}
The task of reversing an unknown quantum channel is interlinked with quantum error cancellation. In quantum information processing, estimating the expectation value $\tr[O\rho]$ for a given observable $O$ and a quantum state plays an essential role~\cite{aaronson2018shadow}. In practice, a state $\rho$ is inevitably affected by noise which is modeled by a quantum channel $\cN$. Consequently, many methods have been proposed to recover $\tr[O\rho]$ against noises rather than obtaining $\tr[O\cN(\rho)]$~\cite{Temme_2017,Suguru2018,Takagi2021}.

One of the primary techniques employed is the \textit{probabilistic error cancellation}~\cite{Temme_2017,Jiang2020} wherein the key idea is to represent the inverse map $\cN^{-1}$ of the noisy channel as a quasi-probabilistic mixture of quantum channels. A crucial assumption in this protocol is that the noise is given a \textit{prior}, otherwise, a high-precision tomography of the noise channel is required. In contrast, the scheme presented in Theorem~\ref{thm:wc_err} shows the potential to achieve high-precision error cancellation for certain unknown channels, e.g., depolarizing channels with unknown parameters within a given range.

In general, any $n$-slot quantum comb $\cC$ can be equivalently realized by a sequence of quantum channels $\{\cE_j\}_{j=1}^{n+1}$ with an ancillary system~\cite{Chiribella_2008a}. Thus, we can obtain a set of channels $\{\cE_{ij}\}_{j}$ for each comb $\cC_i$ in a decomposition of a virtual comb $\VComb = \sum_{i}\eta_i \cC_i$. Given an unknown quantum channel oracle $\cN$ and a noisy state $\cN(\rho)$, if $\VComb(\cN^{\ox n}) = \cN^{-1}$, then we can obtain $\tr[O\rho] = \tr[O\cdot \VComb(\cN^{\ox n})\circ\cN(\rho)]$ by querying $\cN$, sampling quantum channels for each $\cC_i$ and applying classical post-processing~\cite{Jiang2020,Takagi2021}.

In each round out of $S$ times sampling, we sample a sequence of quantum channels $\{\cE_{sj}\}_{j=1}^{n+1}$ from $\{\cE_{ij}\}_{ij}$ with probability $|\eta_s|/\gamma$, where $\gamma = \sum_i |\eta_i|$. Apply $\cE_{s1}, \cN, \cE_{s2}, \cdots, \cN, \cE_{s,n+1}$ to the target state sequentially to obtain
$\cE_{s,n+1}\circ\cN\circ\cdots\circ\cE_{s2}\circ\cN\circ\cE_{s1}\circ\cN(\rho) = \VComb_s(\cN^{\ox n})\circ\cN(\rho)$,
and then measure each qubit on a computational basis. We then denote $\lambda_s$ as the measurement outcome and obtain a random variable $X^{(s)} = \gamma \cdot {\rm sgn}(\eta_s) \lambda_s \in [-\gamma, \gamma]$. After $S$ rounds sampling, we calculate the empirical mean value $\zeta:= \frac{1}{S} \sum_{s=1}^S X^{(s)}$ as an estimation for the expectation value $\tr[O\cdot \VComb(\cN^{\ox n})\circ\cN(\rho)]$. By Hoeffding's inequality~\cite{hoeffding2012collected}, to estimate the expectation value within error $\epsilon$ with probability no less than $1-\delta$, the number of samples required to be $S \geq 2\gamma^2 \log(2/\delta)/\epsilon^2$. Hence, $\gamma$ is known as the sampling overhead. We note that the optimal sampling overhead can be calculated by SDP as given in appendix. 

In particular, if we aim to cancel the effect of an unknown depolarizing noise $\cD_{p}$ from a set of distinct noise parameters as described in Theorem~\ref{thm:main_Ndepo}, we have a detailed protocol provided in appendix, where we do not need to implement quantum combs and instead relies on three types of simple operations: i): do nothing to the received state; ii): replace the received state with a maximally mixed state; iii): apply the black box to the received state iteratively for $i$ times.

\textcolor{black}{\textbf{\emph{Application to universal unitary inversion.---}}}
Now we investigate a particular scenario where the quantum process is known to be a unitary operation. Previously, several works have studied the problem of reversing unknown unitary operations, including deterministic non-exact protocol~\cite{Chiribella_2016,sardharwalla2016universal} and probabilistic exact protocols~\cite{Sedl_k_2019,Quintino_2019a,Quintino2019}. Notably, it is proved that the inverse operation $U^{-1}$ cannot be implemented deterministically and exactly with a single use of $U$~\cite{Chiribella_2016}. Recently, deterministic and exact protocols were proposed, requiring four calls of $U$ in a qubit case~\cite{Yoshida2023}, and $\cO(d^2)$ calls for a general $d$-dimensional $U$~\cite{chen2024quantum}.
In a virtual setting, it is interesting to ask whether a 1-slot virtual comb is enough for reversing an arbitrary unknown $d$-dimensional unitary channel $\cU_d(\cdot) = U_d(\cdot)U_d^\dag$ or not. Here we find the answer is positive as the following result.
\begin{proposition}~\label{prop:uni_chan}
    For any dimension $d$, there exists a $1$-slot virtual comb $\VComb$ that transforms all qudit-unitary channels $\cU_d$ into their inverse $\cU_d^{-1}$, i.e., $\VComb(\cU_d)(\cdot) = U_d^{\dag}(\cdot)U_d$.
\end{proposition}

Proposition~\ref{prop:uni_chan} reveals that with a virtual comb, a deterministic and exact protocol for any dimensions can be achieved with just one call of the unitary. 
This result gives an alternative way to simulate the inverse of unknown unitary in practice with shallower circuits for estimating expectation values. We point out that when there exists depolarizing noise, i.e., the given channel is $\cU_d \circ \cD_{p}$, the 1-slot virtual comb will result in an overall operation as $\cU_d^{-1} \circ \cD_{p}$ and the probabilistic error cancellation could be used to mitigate this error. For the deterministic protocol, the circuit is generally not transversal~\cite{Yoshida2023, chen2024quantum}, thus the depolarizing noise will accumulate and become difficult to handle.

Furthermore, we analyze the query complexity of a virtual protocol, specifically the number of times $U$ needs to be queried to obtain the expectation value $\tr[O\cU_d^{-1}(\rho)]$. The optimal sampling overhead for an $n$-slot virtual comb that can exactly reverse all $d$-dimensional unitaries can be characterized via the following SDP.
\begin{subequations}\label{SDP:unitary_inv}
\begin{align}
\nu(d,n) = \min &\; 2 \eta + 1,\\
 {\rm s.t.} & \; \tr[\VChoi\Omega] = 1, \label{Eq:sdp_exact_inv}\\ 
            & \;\; \VComb = (1+\eta)\cC_0 - \eta \cC_1, \eta \geq 0,\\
            & \;\; \cC_0, \cC_1 \text{ are $n$-slot quantum combs,}
\end{align}
\end{subequations}
where Eq.~\eqref{Eq:sdp_exact_inv} ensures that $\VComb$ is a desired map that can exactly reverse an arbitrary unitary operation and $\Omega$ is a $d^{2(n+1)}\times d^{2(n+1)}$ positive matrix called the performance operator~\cite{Chiribella_2016,Quintino2022}. The detailed formula and numerical results on the sampling overhead for small $d$ and $n$ are provided in appendix.

Notably, we find that the optimal sampling overhead for the virtual comb that can exactly reverse unknown unitary operations has a dual relationship with the problem of finding the optimal average fidelity of reversing unknown unitary operations by a quantum comb~\cite{Quintino2022} as the following theorem. 
\begin{theorem}\label{thm:cost_with_fid}
    The optimal sampling overhead for the $n$-slot virtual comb that can exactly reverse all $d$-dimensional unitary operations satisfies 
    \begin{equation}
        \nu(d,n) = \frac{2}{F_{\rm opt}(d,n)}-1,\quad \forall d\geq 2, n\geq 1,
    \end{equation}
    where $F_{\rm opt}(d,n)$ is the optimal average channel fidelity of reversing all $d$-dimensional unitary operations with an $n$-slot quantum comb.
\end{theorem}

Here, the optimal average channel fidelity is given by $F_{\rm opt}(d,n) = \max \tr[C\Omega]$ where $\Omega$ is the performance operator as appeared in Eq.~\eqref{Eq:sdp_exact_inv} and the maximization ranges over all $n$-slot quantum combs with Choi operators $C$~\cite{Quintino2022}. When $n=1$, it has been shown that $F_{\rm opt}(d, 1) = 2/d^2$~\cite{Chiribella_2016}, leading to $\nu(d,1) = d^2-1$. 
Although the sufficient querying number of the unitary is governed by $1\cdot\nu^2(d,1)$, scaling as $\cO(d^4)$, worse than $\cO(d^2)$ required by the deterministic and exact protocol~\cite{chen2024quantum}, it is worthwhile to note that when the state or observable is given, the query complexity could be significantly reduced. Specifically, we find that to estimate the expectation value $\tr[Z\cU_d^{-1}(\ketbra{0}{0})]$,  the 1-slot virtual protocol has a better performance in both average simulation error and standard deviation under the same number of queries of the unknown unitary. The details of the numerical analysis to show this potential advantage are provided in appendix.

\textcolor{black}{\textbf{\emph{Concluding remarks.---}}}
In this work, we addressed the problem of reversing an unknown quantum process by introducing the \textit{virtual comb}. 
Our theoretical analysis demonstrated its ability and shows its potential to help us further understand the properties and capabilities of channels, combs, and virtual processes. One may already notice that a qubit channel can be determined by 12 parameters, which coincides with our numerical result that if the number of random qubit channels exceeds 13, no perfect 1-slot virtual inversion protocol could be found. 
In terms of applications, the examples we provided suggest that the virtual combs may potentially become an alternative solution in specific experimental settings. It might offer trade-offs in terms of query complexity, circuit depth, and the number of auxiliary qubits; therefore, it is intriguing to conduct further analysis and construct concrete circuits for specific experimental scenarios.

The virtual combs may also shed light on other research directions for unknown processes, particularly in quantum learning. By transmitting quantum states through an unknown process, we can infer its characteristics and replicate it or execute related tasks. Such studies have been done for learning unitary gates~\cite{Bisio_2010,Sedl_k_2019,Mo_2019},  measurements~\cite{Sedl_k_2014,cheng2015learnability}, and Pauli noises~\cite{Chen2023}. How to further extend this setting to learning and using unknown channels remains open. Moreover, virtual combs may also be useful in transforming Hamiltonian dynamics~\cite{Odake2023a,Odake2023}, shadow tomography~\cite{aaronson2018shadow}, virtual resource manipulation~\cite{yuan2023virtual}, and randomized quantum algorithms~\cite{Huang_2020,Wan_2022} for its unique attributes regarding quantum memory effect, sampling, and classical post-processing.

\emph{Acknowledgement.---}
We thank Benchi Zhao, Hongshun Yao, Xuanqiang Zhao, and Kun Wang for their comments. This work was partially supported by the National Key R\&D Program of China (Grant No. 2024YFE0102500), the Guangdong Provincial Quantum Science Strategic Initiative (Grant No. GDZX2303007), the Guangdong Provincial Key Lab of Integrated Communication, Sensing and Computation for Ubiquitous Internet of Things (Grant No. 2023B1212010007), the Start-up Fund (Grant No. G0101000151) from HKUST (Guangzhou), the Quantum Science Center of Guangdong-Hong Kong-Macao Greater Bay Area, and the Education Bureau of Guangzhou Municipality.

\bibliography{main}

\appendix
\setcounter{subsection}{0}
\setcounter{table}{0}
\setcounter{figure}{0}

\vspace{2cm}
\onecolumngrid
\vspace{2cm}

\begin{center}
\large{\textbf{Supplemental Material}}
\end{center}


\renewcommand{\theequation}{S\arabic{equation}}
\renewcommand{\theproposition}{S\arabic{proposition}}
\renewcommand{\thedefinition}{S\arabic{definition}}
\renewcommand{\thefigure}{S\arabic{figure}}
\setcounter{equation}{0}
\setcounter{table}{0}
\setcounter{section}{0}
\setcounter{proposition}{0}
\setcounter{definition}{0}
\setcounter{figure}{0}

In this Supplemental Material, we present detailed proofs of the theorems and propositions in the manuscript ``Reversing Unknown Quantum Processes via Virtual Combs for Channels with Limited Information''. In Appendix~\ref{appendix:preliminary}, we review and derive several useful toolkits for the quantum comb, virtual comb, link product, and Hermitian decomposition of general quantum processes, to make our proofs more self-contained. In Appendix \ref{appendix:proof_of_thm}, we give detailed proofs of the theorems in the manuscript. We also present details about the numerical experiments in the main text. In Appendix~\ref{appendix:protocol}, we provide SDP for calculating the optimal sampling overhead of a virtual comb and a detailed protocol for error cancellation of depolarizing noises. In Appendix~\ref{appendix:connection}, we discuss connections between virtual combs and other quantum methodologies to help understand the virtual comb framework in quantum information processing.

\section{Preliminaries}\label{appendix:preliminary}
\subsection{Comb, virtual comb and link product}

\paragraph{Comb and virtual comb.}
Formally, a quantum comb can be characterized by its Choi operator as the following lemma.
\begin{lemma}[\cite{Chiribella_2008}]\label{lem:comb_def}
Given a matrix $C \in \cL\left(\cP \ox \cI^n \ox \cO^n \ox \mathcal{F}\right)$, it is the Choi operator of a quantum comb $\mathcal{C}$ if and only if it satisfies $C\geq 0$ and
\begin{align}
    C^{(0)}  =1,\; \tr_{\cI_i} [C^{(i)}] =C^{(i-1)} \ox I_{\cO_{i-1}}, \, i =1, \cdots, n+1  \label{Eq:no-sig}
\end{align}
where $C^{(n+1)}:=C, C^{(i-1)}:=\tr_{\cI_i \cO_{i-1}} [C^{(i)}]/d$, $I_{\cH}$ is the identity operator on $\cH$ and $\cI_{n+1}:=\mathcal{F}, \cO_0:=\cP$.
\end{lemma}
Conventionally, a legitimate quantum transformation has to be completely positive (CP) reflected in the positivity of its Choi operator~\cite{choi1975completely}. In the main text, we lift the constraint for a map to be completely positive and introduce the notion of \textit{virtual comb}, which can be decomposed into quantum combs as $\VComb = \eta_0 \mathcal{C}_0 + \eta_1 \mathcal{C}_1$, with $\eta_0$ and $\eta_1$ be arbitrary real numbers. We can easily notice that a virtual comb can be characterized by its Choi operator shown as follows.
\begin{lemma}\label{lem:virtual_comb_choi}
Given a matrix $\VChoi \in \cL\left(\cP \ox \cI^n \ox \cO^n \ox \mathcal{F}\right)$, it is the Choi operator of a virtual comb $\VComb$ if and only if it satisfies
\begin{align}\label{Eq:HPTP_comb}
    \VChoi^{(0)}  =1, \; \tr_{\cI_i} [\VChoi^{(i)}] =\VChoi^{(i-1)} \ox I_{\cO_{i-1}}, \, i =1, \cdots, n+1
\end{align}
where $\VChoi^{(n+1)}:=\VChoi, \VChoi^{(i-1)}:=\tr_{\cI_i \cO_{i-1}} [\VChoi^{(i)}]/d$, $I_{\cH}$ is the identity operator on $\cH$ and $\cI_{n+1}:=\mathcal{F}, \cO_0:=\cP$.
\end{lemma}

\begin{proof}
For the `if part': without loss of generality, we let $\VChoi$ be the Choi operator of a virtual comb $\VComb = \eta_0\cC_0 + \eta_1\cC_1$, where $\eta_0 + \eta_1 =1$. $\VChoi$ satisfies Eq.~\eqref{Eq:HPTP_comb} by the linearity of partial trace. 
For the `Only if' part: let $C = I$, $\eta_1=\|\VChoi\|_1$, and $C' = [(\|\VChoi\|_1 + 1)I - \VChoi]/\|\VChoi\|_1$. If $\VChoi$ satisfies the conditions in Eq.~\eqref{Eq:HPTP_comb}, it is easy to check that $\VChoi = (1+\eta_1) C - \eta_1 C'$, $C,C'\geq 0$ and both satisfy the conditions in Eq.~\eqref{Eq:HPTP_comb}. Hence we complete the proof.
\end{proof}

\paragraph{Link product.}
For any two given processes, they can be connected whenever the input system of one matches the output system of the other. Upon connection, the Choi operator of the overall process can be derived through the \textit{link product} operation~\cite{Chiribella_2008a}, denoted as $*$.
Considering two processes $\cN_{A\rightarrow B}$ and $\cM_{B\rightarrow C}$, the Choi operator of the composite process from $\cH_{A}$ to $\cH_{C}$ can be represented as
\begin{align}
    J_{\cM \circ \cN} = J_{\cN} * J_{\cM} = \tr_B[(J_{\cN} \otimes I_C) \cdot (I_A \otimes J_{\cM}^{T_B})] \, ,
\end{align}
where $J_{\cN}, J_{\cM}$ are the Choi operators of $\cN,\cM$, respectively, and $T_B$ denotes taking partial transpose on $\cH_B$. This representation holds for any processes. It is noteworthy that the link product exhibits both associative and commutative properties:
\begin{align*}
    J_{1} * (J_{2} * J_{3}) &= (J_{1} * J_{2}) * J_{3}\\
    J_{1} * J_{2} &= J_{2} * J_{1}
\end{align*}

\subsection{Hermitian decomposition of virtual channels and combs}\label{appendix:1_decomposition}
To analyze the capability of using a virtual comb to invert general channels and to prove that any two quantum channels can be exactly reversed by a 1-slot virtual comb as stated in Theorem \textcolor{blue}{3}, we apply Hermitian decomposition to represent the Choi operator for general channels and conduct analysis using this formulation.
In this subsection, we introduce the mathematical form of this decomposition and some of its properties. We then show the representation of quantum processes in this form.

For the $d$-dimensional Hilbert space $\cH_d$, denote by $\cL^\dag(\cH_d)$ the set of Hermitian operators on it with a group of orthonormal basis
\begin{equation}
    \cB_d:=\left\{\sigma_0=I_d/\sqrt{d},\sigma_1,\cdots,\sigma_{d^2-1}\right\}\subset\cL^\dag(\cH_d),
\end{equation}
where $\tr[\sigma_j\sigma_k]=\delta_{jk}$ and $\delta_{jk}$ is Kronecker delta. Thus, any $d^n$-dimensional Hermitian operator could be represented as a real linear combination of $\cB_d^{\ox n}$. 

\begin{definition}
Given a $d^n$-dimensional Hilbert space $\cH_{d^n}$, for $\mathbf{j}=\left(j_1,j_2\cdots,j_n\right)\in\{0,1,\cdots,$ $d^2-1\}^n$, denote
\begin{equation}
    \sigma_{\mathbf j}:=
    \begin{cases}
    \sigma_{j_1}\ox\sigma_{j_2}^T\ox\sigma_{j_3}\ox\cdots\ox\sigma_{j_{n}},&\text{if }n\text{ odd};\\
        \sigma_{j_1}^T\ox\sigma_{j_2}\ox\sigma_{j_3}^T\ox\cdots\ox\sigma_{j_{n}},&\text{if }n\text{ even},
    \end{cases}
\end{equation}
then 
\begin{equation}
    \cB_{d,n}:=\left\{\sigma_{\mathbf j}\,\middle|\,\mathbf{j}\in\left\{0,1,\cdots,d^2-1\right\}^n\right\}
\end{equation}
forms a group of orthonormal basis for $\cL^\dagger(\cH_{d^n})$ naturally. 
\end{definition}

\begin{proposition}
For any Hermitian operator $H\in\cL^\dag(\cH_{d^n})$,
\begin{equation}
    H=\sum_{\mathbf j}H_{\mathbf j}\sigma_{\mathbf j},
\text{ where each }
    H_{\mathbf j}=\tr\left[H\cdot\sigma_{\mathbf j}\right]\in\mathbb R.
\end{equation}
\end{proposition}

\begin{definition}
For any Hermitian operator $H\in\cL^\dag(\cH_{d^n})$, denote
\begin{equation}
    M_H:=\sum_{\mathbf j}H_{\mathbf j}\langle\mathbf j\rangle,
\end{equation}
where 
\begin{equation}
    \langle\mathbf j\rangle:=
    \begin{cases}
        \bra{j_1}\ox \ket{j_2}\ox\bra{j_3}\ox\cdots\ox\bra{j_{n}},&\text{if }n\text{ odd};\\        
        \ket{j_1}\ox\bra{j_2}\ox\ket{j_3}\ox\cdots\ox\bra{j_{n}},&\text{if }n\text{ even},
    \end{cases}
\end{equation}
\end{definition}
\begin{fact}
Any state $\rho\in\cL^\dag(\cH_d)$, the Choi matrix of any virtual channel $C_{\widetilde{\cN}}$, the Choi matrix of identity channel and the Choi matrix of any $1$-slot virtual comb $C_\VComb$ could be represented as 
\begin{align}
    &\rho=\sum_j\rho_j\sigma_j,
    &&M_\rho=\sum_j\rho_j\bra{j},\\
    &C_{\widetilde{\cN}}=\sum_{jk}\widetilde{\cN}_{jk}\sigma_j^T\ox\sigma_k, &&M_{\widetilde{\cN}}=\sum_{jk}\widetilde{\cN}_{jk}\ket{j}\ox\bra{k},\\    &C_{\id}=\sum_{j}\sigma_j^T\ox\sigma_j, 
    &&M_{\id}=I_{d^2},\label{Eq:Pauli_channel_id}\\
    &C_\VComb=\sum_{klrs}\widetilde\cC_{klrs}\sigma_k^T\ox\sigma_{l}\ox\sigma_{r}^T\ox\sigma_s,
    &&M_\VComb=\sum_{klrs}\widetilde\cC_{klrs}\ket{k}\ox\bra{l}\ox\ket{r}\ox\bra{s},
\end{align}
where $M_{C_*}$ is abbreviated as $M_*$ when there is no ambiguity.
\end{fact}

\begin{fact}
In the decomposition above, quantum processes could be represented as follows
\begin{align}
    &\widetilde{\cN}(\rho)=\sum_{jk}\rho_{j}\widetilde\cN_{jk}\sigma_k,
    &&M_{\widetilde{\cN}(\rho)}=M_\rho\cdot M_{\widetilde\cN},\\
    &C_{\widetilde{\cN}\circ\widetilde{\cN}'}=\sum_{jkl}\widetilde\cN'_{jk}\widetilde\cN_{kl}\sigma_j^T\ox\sigma_l,
    &&M_{\widetilde{\cN}\circ \widetilde{\cN}'}=M_{\widetilde\cN'}\cdot M_{\widetilde\cN},\label{Eq:Pauli_channel_prod}\\
    &C_{\VComb(\widetilde{\cN})}=\sum_{klrs}\widetilde\cC_{klrs}\widetilde\cN_{lr}\sigma_k^T\ox\sigma_s,
    &&M_{\VComb(\widetilde{\cN})}
    =\left(I_{d^2}\ox M_{\widetilde\cN}^{T_1}\right)\cdot M_{\widetilde\cC}^{T_2},\label{Eq:Pauli_comb_link_channel}
\end{align}
where 
\begin{equation}
    M_{\widetilde\cN}^{T_1}:=\sum_{jk}\widetilde{\cN}_{jk}\bra{j}\ox\bra{k},\, 
    M_{\widetilde\cC}^{T_2}:=\sum_{klrs}\widetilde\cC_{klrs}\ket{k}\ox\ket{l}\ox\ket{r}\ox\bra{s}.
\end{equation}
\end{fact}
This fact can be directly checked by calculating
\begin{align}
    \widetilde{\cN}(\rho)&=\sum_{lr}\sum_{j'}\rho_{j'}\braandket{l}{\sigma_{j'}}{r}\sum_{jk}\widetilde\cN_{jk}\braandket{l}{\sigma_{j}^T}{r}\sigma_k
    =\sum_{j'jk}\rho_{j'}\widetilde\cN_{jk}\sigma_k\sum_{lr}\braandket{l}{\sigma_{j'}}{r}\braandket{r}{\sigma_{j}}{l}\\
    &=\sum_{j'jk}\rho_{j'}\widetilde\cN_{jk}\sigma_k\tr[\sigma_{j'}\sigma_{j}]=\sum_{jk}\rho_{j}\widetilde\cN_{jk}\sigma_k.
\end{align}
Other equations could be obtained analogously.
Based on these symbolic expressions, we have the following statement as a direct corollary of Eq.~\eqref{Eq:Pauli_channel_id} and Eq.~\eqref{Eq:Pauli_channel_prod}.

\begin{corollary}
For a channel $\widetilde{\cN}$, it has inverse map $\widetilde{\cN}^{-1}$ if and only if $M_{\widetilde{\cN}}$ is invertible, while 
\begin{equation}
    M_{\widetilde{\cN}^{-1}}=M_{\widetilde{\cN}}^{-1}.
\end{equation}
\end{corollary}

\begin{lemma}\label{lemma:Pauli_TP}
For a given matrix $M_{\widetilde{\cN}}\in\mathbb R^{d^2\times d^2}$, $C_{\widetilde{\cN}}$ is the Choi matrix of an HPTP map if and only if $\widetilde\cN_{j0}=\delta_{j0}$; for a given matrix $M_\VComb\in\mathbb R^{d^4\times d^4}$, $C_\VComb$ is the Choi matrix of a $1$-slot virtual comb if and only if 
\begin{equation}
    \widetilde\cC_{klr0}=\begin{cases}
        0, &\text{if }r\ne0;\\
        \delta_{k0}, &\text{if }l=r=0.
    \end{cases}
\end{equation}
\end{lemma}
\begin{proof}
By 
\begin{equation}
    \forall \rho,\ \tr[\rho]=\sqrt{d}\rho_0,\ \tr[\widetilde{\cN}(\rho)]=\sum_{jk}\rho_j\widetilde\cN_{jk}\tr[\sigma_k]=\sqrt{d}\sum_{j}\rho_j\widetilde\cN_{j0},
\end{equation}
we find $\widetilde{\cN}$ is TP if and only if $\widetilde\cN_{j0}=\delta_{j0}$.
By Lemma \ref{lem:virtual_comb_choi}, we have $C_\VComb$ is the Choi matrix of a $1$-slot virtual comb if and only if 
\begin{equation}
    \tr_4 [C_\VComb]=\tr_{34}[C_\VComb\ox I_d/d],\ 
    \tr_{234} [C_\VComb]={\tr [C_\VComb]}\cdot I_d/d,\ 
    \tr [C_\VComb]=d^2,
\end{equation}
i.e.
\begin{equation}
    (r\ne0\implies \widetilde\cC_{klr0}=0)\wedge
    \widetilde\cC_{k000}=\delta_{k0}  
\end{equation}
\end{proof}

\section{Proof of Theorems}\label{appendix:proof_of_thm}

\subsection{Reversing depolarizing channels}\label{appendix:depo}
\renewcommand\theproposition{\textcolor{blue}{1}}
\setcounter{proposition}{\arabic{proposition}-1}
\begin{theorem}
    For any $n\geq 1$, let $\cD_{p_1},..., \cD_{p_{n+1}}$ be $n+1$ $d$-dimensional depolarizing channels with distinct noise parameters $p_1,...,p_{n+1} \in[0,1)$. There exists an $n$-slot virtual comb $\VComb$ satisfying
    \begin{equation}\label{Eq:appendix_main}
    \VComb(\cD_{p_i}^{\ox n}) = \cD_{p_i}^{-1}, \quad \forall \, i=1,...,n+1.
    \end{equation}
\end{theorem}

\begin{proof}
We first prove the case for $n=1$, and show the main idea. Denote the Choi operator of the virtual comb as $\VChoi_{\cP\cI\cO\cF}$ and the Choi operator of a qudit depolarizing channel $\cD_{p}$ as:
\begin{align}
    J_{\cD_p} = (1-p)J_{\id,\cI\cO} + p J_{\cD,\cI\cO} \, ,
\end{align}
where $J_{\id}=\dketbra{I_d}{I_d}$ and $J_{\cD} = \frac{1}{d}I_d \ox I_d$ are the Choi operators of the identity channel and fully depolarizing channel respectively, and $\cP$, $\cI$, $\cO$, $\cF$ represent the corresponding systems in Fig \textcolor{blue}{1} in the main text. By direct calculation, the Choi operator of $\cD_p^{-1}$ can be written as
\begin{align}\label{Eq:JDp_inverse}
    J_{\cD_p^{-1}} = \dfrac{1}{1-p} J_{\id} - \dfrac{p}{1-p} J_{\cD} \, .
\end{align}
Notice that $[J_{\cD_p^{-1}}, U\otimes \overline{U}]=0$ for arbitrary $U$ in $SU(d)$, which means $\cU^\dag \circ \cD_{p}^{-1} \circ \cU = \cD_{p}^{-1}$. Thus for arbitrary virtual comb $\VComb$ satisfying Eq.~\eqref{Eq:appendix_main}, $\VComb' = \int dU \, \cU^\dag \circ \VComb \circ \cU$ is also a feasible one. 
As quantum channels from $\cL(\cP)$ to $\cL(\cF)$ with Choi operator commute with $U_\cP\otimes \overline{U}_\cF$ could always be decomposed into a linear combination of the identity channel and the fully depolarizing channel, 
without loss of generality, we could focus on a virtual comb satisfying $[\VChoi, U_\cP\otimes \overline{U}_\cF \otimes I_{\cI\cO}]=0$ and get
\begin{align}\label{equation_set_1}
    \left\{
    \begin{aligned}
        \VComb(\id_{\cI\cO}) = \alpha \id_{\cP\cF} + (1-\alpha)\cD_{\cP\cF}\\
        \VComb(\cD_{\cI\cO}) = \beta \id_{\cP\cF} + (1-\beta)\cD_{\cP\cF}\\
    \end{aligned}
     \right.
\end{align}
Inserting Eq.(\ref{equation_set_1}) into Eq.~\eqref{Eq:appendix_main} gives us the following linear equations with variables $\alpha,\beta$.
\begin{equation}\label{Eq:comb_alpha_beta}
    \left\{
    \begin{aligned}
    &(1-p_1)\alpha + p_1\beta = \frac{1}{1-p_1} \\
    &(1-p_2)\alpha + p_2\beta = \frac{1}{1-p_2}
    \end{aligned}
    \right.
    \implies
    \alpha = \frac{1-p_1-p_2}{(1-p_1)(1-p_2)}, \beta = \frac{2-p_1-p_2}{(1-p_1)(1-p_2)}.
\end{equation}
Therefore, we have that the linear map satisfying Eq.~\eqref{equation_set_1} with $\alpha$, $\beta$ given by Eq.~\eqref{Eq:comb_alpha_beta} could exactly reverse $\cD_{p_1}$ and $\cD_{p_2}$ using one-call of the channel.

Now the last step is to show that the linear map derived for two distinct depolarizing channels could always be realized with a 1-slot virtual comb. Consider three quantum combs $\cC_1,\cC_2,\cC_3$ satisfying the following condition for arbitrary $\cN \in \CPTP(\cI,\cO)$ 
\begin{equation}\label{n1_comb_decomposition}
    \cC_1(\cN) =  \cD, \; \cC_2(\cN) = \id, \; \cC_3(\cN) = \cN \, .
\end{equation}
It is easy to see that $\cC_1$ corresponds to always passing through a fully depolarizing channel, $\cC_2$ corresponds to bypassing channel $\cN$ and the input state is directly output from $\cP$ to $\cF$, and $\cC_3(\cN)$ corresponds to go through channel $\cN$.
For any given $\alpha, \beta\in\mathbb{R}$, we could then construct a virtual comb $\VComb$ such that 
\begin{equation}\label{Eq:vcomb_construct}
    \VComb(\cN) = \beta\cC_2(\cN) + (\alpha-\beta)\cC_3(\cN) + (1-\alpha)\cC_1(\cN) \, .
\end{equation}
It is easy to check that $\VComb$ satisfies Eq.~\eqref{equation_set_1} when the input $\cN$ is chosen to be $\cD$ and $\id$, respectively.

    For general $n$, we use a similar method to demonstrate the existence of a virtual comb as the example for $n=1$ given in the main text. Let $\mathbf{k} = (k_1, k_2, ..., k_n)\in \{0,1\}^n$ be a $n$-bit binary string and $|\mathbf{k}|$ be the Hamming weight of it. We then denote $\cD_{\mathbf{k}} = \cN_{k_1}\ox \cN_{k_2}\ox \cdots \ox \cN_{k_n}$ where $\cN_{k_i} = \id$ if $k_i=1$ and $\cN_{k_i} = \cD$ if $k_i = 0$. Suppose $\VComb$ is a virtual comb that could do exact inversion for $m$ distinct noise level, 
    it is then sufficient to satisfy
    \begin{equation}\label{appendEq:Nslot_inverse}
        \VComb * \cD_{p_i}^{\ox n} = \sum_{\mathbf{k}\in \{0,1\}^{N}} (1-p_i)^{|\mathbf{k}|}p_i^{1-|\mathbf{k}|} \VComb * \cD_{\mathbf{k}} = \frac{1}{1-p_i} \id_{\cP\cF} - \frac{p_i}{1-p_i} \cD_{\cP\cF}, \;\; \forall i=1,2,\cdots,m \, .
    \end{equation}
    We denote $\VComb * \cD_{\mathbf{k}} = \alpha_{\mathbf{k}}\id_{\cP\cF} + (1-\alpha_{\mathbf{k}}) \cD_{\cP\cF}$ and $\alpha_{i} = \sum_{|\mathbf{k}|=N-i} \alpha_{\mathbf{k}}$. Then we can equivalently express Eq.~\eqref{appendEq:Nslot_inverse} as the following linear system
    \begin{equation}\label{appendEq:nslot_coeff}
    \left\{\begin{array}{c}
    \left(1-p_1\right)^n \alpha_0+\left(1-p_1\right)^{n-1} p_1 \alpha_1+\cdots+p_1^n \alpha_n=\frac{1}{1-p_1} \\
    \left(1-p_2\right)^n \alpha_0+\left(1-p_2\right)^{n-1} p_2 \alpha_1+\cdots+p_2^n \alpha_n=\frac{1}{1-p_2} \\
    \quad \quad\vdots\\
    \left(1-p_m\right)^n \alpha_0+\left(1-p_m\right)^{n-1} p_m \alpha_1+\cdots+p_m^n \alpha_n=\frac{1}{1-p_m}
    \end{array}\right.
    \end{equation}
    The coefficient matrix of the linear system in Eq.~\eqref{appendEq:nslot_coeff} can be written as 
    \begin{equation}
    A = \left(\begin{array}{cccc}
    (1-p_1)^n & 0 & \ldots & 0 \\
    0 & (1-p_2)^n & \ldots & 0 \\
    \vdots & \vdots & \ddots & \vdots \\
    0 & 0 & \ldots & (1-p_m)^n
    \end{array}\right)
    \left(\begin{array}{cccc}
    1 & \frac{p_1}{1-p_1} & \ldots & (\frac{p_1}{1-p_1})^{n} \\
    1 & \frac{p_2}{1-p_2} & \ldots & (\frac{p_2}{1-p_2})^{n} \\
    \vdots & \vdots & \ddots & \vdots \\
    1 & \frac{p_m}{1-p_m} & \ldots & (\frac{p_m}{1-p_m})^{n}
    \end{array}\right).
    \end{equation}
    Notice that the right one is a Vandermonde matrix. When $m=n+1$, its determinant can be expressed as
    \begin{equation}
        \det(A) = \prod_{1\leq i\leq m} (1-p_i)^n \prod_{1\leq i< j\leq m} \left(\frac{p_j}{1-p_j} - \frac{p_i}{1-p_i}\right)
    \end{equation}
    which indicates that it is invertible if and only if all $p_i$ are distinct. Hence, we conclude that $\rank A =\min\{n+1,m\}$. 
    Now we consider the augmented matrix of the linear system in Eq.~\eqref{appendEq:nslot_coeff}. We have
    \begin{equation}
        (A|B) = 
        \left(\begin{array}{ccccc}
        (1-p_1)^n & (1-p_1)^{n-1}p_1 & \ldots & p_1^n & \frac{1}{1-p_1}\\
        (1-p_2)^n & (1-p_2)^{n-1}p_2 & \ldots & p_2^n & \frac{1}{1-p_2} \\
        \vdots & \vdots & \ddots & \vdots & \vdots\\
        (1-p_m)^n & (1-p_m)^{n-1}p_m & \ldots & p_m^n & \frac{1}{1-p_m}
        \end{array}\right).
    \end{equation}
    By applying row addition, it is easy to verify that the determinant of $(A|B)$ is proportional to the following matrix
    \begin{equation}
        \left(\begin{array}{ccccc}
        (1-p_1)^{n+1} & (1-p_1)^{n} & \ldots & (1-p_1) & 1\\
        (1-p_2)^{n+1} & (1-p_2)^{n} & \ldots & (1-p_2) & 1 \\
        \vdots & \vdots & \ddots & \vdots & \vdots\\
        (1-p_m)^{n+1} & (1-p_m)^{n} & \ldots & (1-p_m) & 1
        \end{array}\right)
    \end{equation}
    which is also a Vandermonde matrix. Using a similar argument, we have $\rank(A|B)=\min\{n+2,m\}$. When $m=n+1$, it follows $\rank A = \rank(A|B)$ which yields that the linear system has a unique solution. Similar to Eq.~\eqref{n1_comb_decomposition}, we can consider quantum combs $\cC_{\id},\cC_{\cD},\cC_i$ that satisfy the following
    \begin{equation}
        \cC_{\id}(\cN^{\ox n}) = \id,\; \cC_{\cD}(\cN^{\ox n}) = \cD,\; \cC_i(\cN^{\ox n}) = \cN^{i},\; \forall \, \cN\in \CPTP(\cI, \cO), \, i=1,2,...,n, 
    \end{equation}
    where $\cN^{i}$ here means sequentially passing through channel $\cN$ for $i$ times. 
    Then, the whole virtual comb satisfying Eq.~\eqref{appendEq:Nslot_inverse} could be constructed as $\VComb = \eta_{id} \cdot \cC_{id} + \eta_{\cD}\cdot \cC_{\cD} + \sum_{1}^{n} \eta_{i} \cdot \cC_i$, where the coefficients are given by the solution of the linear equations in Eq.~\eqref{appendEq:nslot_coeff}.
    Hence we complete the proof.
\end{proof}
\renewcommand{\theproposition}{S\arabic{proposition}}

From the proof above, one could immediately derive the Theorem \textcolor{blue}{2} in the main text that an $n$-slot virtual comb is not able to do exact inversion for depolarizing noise with $n+2$ distinct noise levels.
The main idea is that the coefficient matrix and the augmented matrix above are both Vandermonde matrices, calculating their rank shows that $m \geq n+2$ will result in an inequality which makes the linear system in Eq.~\eqref{appendEq:nslot_coeff} have no solution. To highlight the unique power of reversing a family of depolarizing channels with unknown noises provided by virtual combs, we show that such an exact channel inversion task cannot be accomplished via a quantum comb probabilistically as the following proposition.

\renewcommand\theproposition{\textcolor{blue}{2}}
\setcounter{proposition}{\arabic{proposition}-1}
\begin{proposition}
    For any $n \geq 1$, let $\cD_{p_1},..., \cD_{p_{n+1}}$ be $n+1$ $d$-dimensional depolarizing channels with distinct noise parameters $p_1,...,p_{n+1} \in[0,1)$. There does not exist an $n$-slot quantum comb that can reverse $\cD_{p_i}$ for each $i$, even probabilistically.
\end{proposition}
\begin{proof}
    For $n=1$, suppose there is a quantum comb $\cC$ with a Choi operator $C$ that can probabilistically reverse an unknown depolarizing channel $\cD_p$. It follows
    \begin{equation}\label{Eq:comb_prob}
    \left\{\begin{array}{c}
        C * \dketbra{I_d}{I_d}_{\cI\cO} = \alpha_1 \dketbra{I_d}{I_d}_{\cP\cF} + \alpha_2 (\frac{1}{d}I_d \ox I_d)_{\cP\cF}\\
        C * (\frac{1}{d}I_d \ox I_d)_{\cI\cO} = \beta_1 \dketbra{I_d}{I_d}_{\cP\cF} + \beta_2 (\frac{1}{d}I_d \ox I_d)_{\cP\cF}
    \end{array}
    \right.
    \end{equation}
    where $\alpha_{1,2}, \beta_{1,2} \geq 0, \alpha_1 + \alpha_2 \leq 1, \beta_1 + \beta_2 \leq 1$. At the same time, for $i=1,2$, we have
    \begin{equation}\label{Eq:comb_prob_q}
    \left\{\begin{array}{c}
        (1-p_i)\alpha_1 + p_i \beta_1 = \frac{q}{1-p_i}\\
        (1-p_i)\alpha_2 + p_i \beta_2 = \frac{-p_iq}{1-p_i}
    \end{array}
    \right.
    \end{equation}
    where $0<q\leq 1$. Notice that $1-p_i\geq 0,\alpha_2\geq 0,p_i\geq 0\,\beta_2,\geq 0$ and $-p_iq/(1-p_i)\leq 0$ which is a contradiction in Eq.~\eqref{Eq:comb_prob_q}. Thus, there is no probabilistic comb. 

    For $n \geq 2$, the proof is similar, and the key point is that in Eq.~\eqref{appendEq:Nslot_inverse} if the virtual comb is replaced with a quantum comb or a probabilistic comb, the left-hand side of the equation is always positive while the right-hand side has negative eigenvalues.
\end{proof}
\renewcommand{\theproposition}{S\arabic{proposition}}

The proof of Theorem \textcolor{blue}{4} is based on a specific protocol we provided, where we constructed an $n$-slot virtual comb that can exactly reverse a depolarizing channel whose noise parameter is from $\big\{p_1+(p_2-p_1)k/n\big\}_{k=0}^{n}$. It can be demonstrated that this protocol has a worst-case error of at most $\cO(n^{-1})$ for any noise levels between two parameters that can be exactly reversed.
\renewcommand\theproposition{\textcolor{blue}{4}}
\setcounter{proposition}{\arabic{proposition}-1}
\begin{theorem}
    Let $0\leq p_1<p_2\leq 1$, the minimum worst-case error of an approximate channel inversion for a depolarizing channel $\cD_{p}$ with $p\in [p_1, p_2]$ using an $n$-slot virtual comb is at most $\cO(n^{-1})$.
\end{theorem}
\begin{proof}
Fix $n\geq 2$, by Theorem \textcolor{blue}{1}, we consider a virtual comb that holds $\VChoi * J_{\cD_{p_i}}^{\ox n} = J_{\cD_{p_i}}^{-1}$ for any $p_1,p_2,...,p_{n+1}$ where $p_i = p_1 + \frac{i-2}{n}(p_2-p_1), i\geq 3$. It follows
\begin{equation*}\label{appendEq:nslot_dist}
    \left\| J_{\cD_p} * (\VChoi * J_{\cD_p}^{\ox n}) - \dketbra{I}{I} \right\|_1 = \frac{3}{2}\left|1-(1-p)\cdot\left[(1-p)^n\alpha_0^{(n)} + (1-p)^{n-1}p\alpha_1^{(n)} + \cdots + p^n\alpha_n^{(n)}\right]\right|
\end{equation*}
Denote $f_n(p) = 1-(1-p)\cdot\left[(1-p)^n\alpha_0^{(n)} + (1-p)^{n-1}p\alpha_1^{(n)} + \cdots + p^n\alpha_n^{(n)}\right]$ which is an univariate polynomial of degree $n+1$. Its root set is provided by $\{p_i\}_{i=1}^{n+1}$. Noticing $f_n(0) = 1-\alpha_0^{(n)}$, we can rewrite $f_n(p)$ as
\begin{equation}
    f_n(p) = (-1)^{n+1}\left(1-\alpha_0^{(n)}\right)\cdot\prod_{i=1}^{n+1} \frac{p-p_i}{p_i}.
\end{equation}
Consider the derivatives of this polynomial function at $p_j$. We have
\begin{equation}
    \frac{\mathrm{d} f_n}{\mathrm{d} p}\bigg|_{p=p_i} = (-1)^{n+1}\frac{1-\alpha_0^{(n)}}{p_i}\cdot\prod_{j\neq i}^{n+1} \frac{p_i-p_j}{p_j}.
\end{equation}
Recall that by Eq.~\eqref{appendEq:nslot_coeff}, we have $\alpha_0^{(n)} = 1-(-1)^{n+1}\prod_{i=1}^{n+1}\frac{p_i}{1-p_i}$. In the following, we will show that 
\begin{equation}
    \mathrm{abs}\left(\frac{\mathrm{d} f_n}{\mathrm{d} p}\bigg|_{p=p_1}\right) \leq \mathrm{abs}\left(\frac{\mathrm{d} f_1}{\mathrm{d} p}\bigg|_{p=p_1}\right), \;\mathrm{abs}\left(\frac{\mathrm{d} f_n}{\mathrm{d} p}\bigg|_{p=p_k}\right) \leq \mathrm{abs}\left(\frac{\mathrm{d} f_n}{\mathrm{d} p}\bigg|_{p=p_1}\right), \quad \forall n, k\geq 2.
\end{equation} 
Firstly, we can calculate that
\begin{equation}\label{appendEq:df_n/df_1}
\begin{aligned}
    \mathrm{abs}\left(\frac{\mathrm{d} f_n}{\mathrm{d} p}\bigg|_{p=p_1} \right) \bigg/ \mathrm{abs}\left(\frac{\mathrm{d} f_1}{\mathrm{d} p}\bigg|_{p=p_1}\right) &= \frac{1-\alpha_{0}^{(n)}}{1-\alpha_0^{(1)}}\cdot\prod_{j=3}^{n+1}\frac{p_j-p_1}{p_j}\\
    &= \prod_{i=3}^{n+1}\frac{p_i}{1-p_i}\cdot \prod_{j=3}^{n+1}\frac{p_j-p_1}{p_j}\\
    &= \prod_{j=3}^{n+1}\frac{p_j-p_1}{1-p_j}\\
    &\leq 1
\end{aligned}
\end{equation}
Here we note that $0\leq p_1<p_2\leq 1$ and $p_{i+1}-p_i = p_{i+2}-p_{i+1}$ for any $i$. It holds that
\begin{equation}
    \frac{(p_1-p_i)(p_1-p_{n+4-i})}{(1-p_i)(1-p_{n+4-i})} \leq 1,\quad \forall i\geq 3,
\end{equation}
which yields the inequality in Eq.~\eqref{appendEq:df_n/df_1}. Secondly, we have
\begin{equation}
    \mathrm{abs}\left(\frac{\mathrm{d} f_n}{\mathrm{d} p}\bigg|_{p=p_k} \right) \bigg/ \mathrm{abs}\left(\frac{\mathrm{d} f_n}{\mathrm{d} p}\bigg|_{p=p_1}\right) = \bigg|\frac{\prod_{j\neq k}(p_k-p_j)}{\prod_{j\neq 1}(p_1-p_j)}\bigg| \leq 1.
\end{equation}
Therefore, for a given noise region $[p_1,p_2]$, we denote $c=\frac{\mathrm{d} f_1}{\mathrm{d} p}|_{p=p_1}$ and have 
\begin{equation}\label{appendEq:wost_case_scalling}
    f_n(p)\leq |c|\cdot \frac{p_2-p_1}{n} \, .
\end{equation}
As shown in~\cite{watrous2018theory} the worst-case diamond norm error can be bounded by the entanglement error with a relation that
\begin{align}
    \frac{1}{2}\left\| \cN_1 - \cN_2 \right\|_\diamond \leq d \cdot \frac{1}{2}\left\| J_{\cN_1} - J_{\cN_2} \right\|_1 \, ,
\end{align}
Eq.~\eqref{appendEq:wost_case_scalling} indicates the worst-case error is at most $\cO(n^{-1})$.
\end{proof}
\renewcommand{\theproposition}{S\arabic{proposition}}

With the virtual comb that can exactly reverse the depolarizing channel with the noise level in $\big\{p_1+(p_2-p_1)k/n\big\}_{k=0}^{n}$, 
when the actual noise level of the depolarizing channel is $p\in[p_1,p_2]$, the error is an `$n+1$'-degree function of $p$ in the formula $e^n(p) = \xi \prod_{i}(p-p_i')$, where $\xi$ is a coefficient determined by $\big\{p_1+(p_2-p_1)k/n\big\}_{k=0}^{n}$. 
Then we could calculate the worst-case error by finding the extreme points of different polynomials. 
Based on this, we present the upper bounds of the minimum worst-case error for the cases $p_1, p_2$ are $(0,0.2)$, $(0,0.4)$ and $(0,0.6)$ in Fig. \textcolor{blue}{2}.

\subsection{SDP formula fo approximate channel inversion}\label{appendix:appro}
In this subsection, we remark why the optimization of the average error for approximately reversing quantum channels within a channel set $\Theta$ can be determined via an SDP. Notice that the involvement of a virtual comb may render the entire process not necessarily a quantum channel, we quantify the performance by considering the completely bounded trace distance from $\VComb(\cN_i^{\ox n}) \circ \cN_i$ to the identity channel `$\id$'. Hence, the average error for approximately reversing quantum channels within a channel set $\Theta$ can be defined as follows.
\begin{subequations}
\begin{align}
e_{\rm ave,opt}^n(\Theta) := \min &\; \frac{1}{2}\sum_{i=1}^{m} p_i \left\| \VComb(\cN_i^{\ox n}) \circ \cN_i - \id \right\|_\diamond,\\
 {\rm s.t.} & \;\; \VComb = (1+\eta) \cC_0 - \eta \cC_1, \eta\geq 0,\\
            & \;\; \cC_0, \cC_1 \text{ are $n$-slot quantum combs.}
\end{align}
\end{subequations}
For any two HPTP maps $\cN_1,\cN_2$ from system $A$ to $B$, the completely bounded trace distance can be evaluated using the following SDP~\cite{watrous2009semidefinite}, which is a minimization problem.
\begin{subequations}\label{SDP:diamondnorm}
\begin{align}
\frac{1}{2}\|\cN_1-\cN_2\|_{\diamond} = \min &\; \mu\\
 {\rm s.t.} & \;\; Z_{AB}\geq 0, \tr_B [Z_{AB}] \leq \mu I_A,\\
            & \;\; Z_{AB}\geq J_{\cN_1} - J_{\cN_2}.
\end{align}
\end{subequations}
Therefore, given a quantum channel ensemble $\{(p_i, \cN_i)\}_{i=1}^m$, the optimal average error for reversing the channel set $\Theta$ can be expressed as
\begin{subequations}\label{SDP:err_ave}
\begin{align}
e_{\rm ave,opt}^n(\Theta) = \min &\; \sum_{i=1}^{m} p_i \mu_i,\\
 {\rm s.t.} & \;\; \widetilde{C} = (1+\eta) C_0 - \eta C_1, \eta\geq 0, \label{Eq:etaC}\\
            & \;\; C_0, C_1 \text{ are Choi operators of $n$-slot quantum combs,} \label{Eq:comb_cons}\\
            & \;\; Z^{(i)}_{AB}\geq 0, \, \tr_B [Z^{(i)}_{AB}] \leq \mu_i I_A,\\
            & \;\; Z^{(i)}_{AB} \geq \widetilde{C} * J_{\cN_i}^{\ox n+1} - J_{\id}, \forall i,
\end{align}
\end{subequations}
where $J_{\cN_i}$ is the Choi operator of the channel $\cN_i$. We note that Eq.~\eqref{Eq:etaC} is not a valid linear constraint for SDP as $\eta$ and $C_0$ are both variables. However, we can rewrite the constraint as $\widetilde{C} = C_0 - C_1$ and modify the trace constraints in Eq.~\eqref{Eq:comb_cons} by $\tr C_0 = (1+\eta)d^2$ and $\tr C_1 = \eta d^2$. Thus, the resulting optimization problem is a valid SDP.

\subsection{Reversing general channels}\label{appendix:genral_chan}
The proof of Theorem \textcolor{blue}{3} in the main text is based on the Hermitian decomposition technique which is introduced in Appendix~\ref{appendix:preliminary}. Before proving the theorem, we will first prove several lemmas and corollaries as follows. For simplicity, we label the systems $\cP,\cI,\cO,\cF$ as $1,2,3,4$, respectively.

\begin{lemma}
For a $1$-slot virtual comb $\VComb$ and a channel $\cN$, $\VComb(\cN)$ is the inverse map for channel $\cN$, if and only if
\begin{equation}
    \left(M_\cN\ox M_{\cN}^{T_1}\right)\cdot M_{\VComb}^{T_2}=I_{d^2}.
\end{equation}
\end{lemma}
\begin{proof}
$\VComb*\cN$ is the inverse map for channel $\cN$ if and only if $(\VComb*\cN)\circ\cN$ is just identity map. By \eqref{Eq:Pauli_channel_id}, \eqref{Eq:Pauli_channel_prod} and \eqref{Eq:Pauli_comb_link_channel}, we have
\begin{align}
    I_{d^2}=M_{\id}
    =M_{\VComb(\cN)\circ\cN}
    =M_\cN\cdot M_{\VComb(\cN)}
    =M_\cN\cdot\left(\left(I_{d^2}\ox M_{\cN}^{T_1}\right)\cdot M_{\VComb}^{T_2}\right)
    =\left(M_\cN\ox M_{\cN}^{T_1}\right)\cdot M_{\VComb}^{T_2}
\end{align}
\end{proof}
\begin{corollary}\label{Cor:comb_inverse_LS}
    For a $1$-slot virtual comb $\VComb$ and several channels $\cN_0,\cN_1,\cdots,\cN_{m-1}$, each $\VComb(\cN_j)$ is the inverse map for channel $\cN_j$, if and only if
\begin{equation}
    \left(\sum_{j=0}^{m-1}\ket{j}\ox M_{\cN_j}\ox M_{\cN_j}^{T_1}\right)\cdot M_{\VComb}^{T_2}=\sum_{j=0}^{m-1}\ket{j}\ox I_{d^2}.
\end{equation}
\end{corollary}
\begin{lemma}\label{Lem:2d2_rank}
For two different channels $\cN_0,\cN_1$, if they both have inverse map, then
\begin{equation}
 \rank\left(\sum_{j=0}^1\ket{j}\ox M_{\cN_j}\ox M_{\cN_j}^{T_1}\right)=2d^2.
\end{equation}
\end{lemma}

\begin{proof}
Since such matrix is of dimension $2d^2\times d^6$, it is of full rank if and only if it is row linearly independent, i.e. 
\begin{equation}
    \sum_{k=0}^{d^2-1}\sum_{j=0}^1a_{kj}\bra{kj}\cdot\sum_{j=0}^1\ket{j}\ox M_{\cN_{j}}\ox M_{\cN_{j}}^{T_1}=\sum_{j=0}^1\sum_{k=0}^{d^2-1}a_{kj}\bra{k}M_{\cN_{j}}\ox M_{\cN_{j}}^{T_1}=0 \implies \forall\,k,j,\ a_{kj}=0.
\end{equation}
Supposed that the left hands holds, considering $M_{\cN_{0}}^{T_1}$ is linear independent from $M_{\cN_{1}}^{T_1}$, we derive that
\begin{equation}
    \forall\,j,\ \sum_{k=0}^{d^2-1}a_{kj}\bra{k}M_{\cN_{j}}=0.
\end{equation}
Moreover, by each $M_{\cN_{j}}$ invertible, we conclude that each $a_{kj}$ vanishes, i.e. such matrix is of full rank.
\end{proof}
\begin{corollary}\label{Cor:2_channel_inverse_mat}
For two different channels $\cN_0,\cN_1$, if they both have inverse maps, then there always exists $M_{\widetilde \cC'}\in\mathbb R^{d^6\times d^2}$ satisfying
\begin{equation}
    \left(\sum_{j=0}^1\ket{j}\ox M_{\cN_j}\ox M_{\cN_j}^{T_1}\right)\cdot M_{\widetilde \cC'}^{T_2}=\sum_{j=0}^1\ket{j}\ox I_{d^2}.
\end{equation}
\end{corollary}
\begin{proof}
By Lemma \ref{Lem:2d2_rank}, the rows of matrix $\sum_{j=0}^1\ket{j}\ox M_{\cN_j}\ox M_{\cN_j}^{T_1}$ are of full rank, so it has right inverse $V\in\mathbb R^{d^6\times 2d^2}$, i.e.
\begin{equation}
    \left(\sum_{j=0}^1\ket{j}\ox M_{\cN_j}\ox M_{\cN_j}^{T_1}\right)\cdot V=I_{2d^2}.
\end{equation}
Then let
\begin{equation}
    M_{\widetilde \cC'}^{T_2}:=V\cdot\left(\sum_{j=0}^1\ket{j}\ox I_{d^2}\right),
\end{equation}
and we find
\begin{equation}
    \left(\sum_{j=0}^1\ket{j}\ox M_{\cN_j}\ox M_{\cN_j}^{T_1}\right)\cdot M_{\widetilde \cC'}^{T_2}
    =I_{2d^2}\cdot\left(\sum_{j=0}^1\ket{j}\ox I_{d^2}\right)=\sum_{j=0}^1\ket{j}\ox I_{d^2}.
\end{equation}
\end{proof}
\begin{corollary}\label{Cor:2_channel_inverse_comb}
For two different channels $\cN_0,\cN_1$, if they both have inverse maps, then there always exists a virtual comb $\VComb$ satisfying
\begin{equation}\label{Eq:1slot_inverse_LS}
    \left(\sum_{j=0}^1\ket{j}\ox M_{\cN_j}\ox M_{\cN_j}^{T_1}\right)\cdot M_\VComb^{T_2}=\sum_{j=0}^1\ket{j}\ox I_{d^2}.
\end{equation}
\end{corollary}
\begin{proof}
By Corollary \ref{Cor:2_channel_inverse_mat}, there always exists $M_{\VComb'}\in\mathbb R^{d^6\times d^2}$ satisfying
\begin{equation}
    \left(\sum_{j=0}^1\ket{j}\ox M_{\cN_j}\ox M_{\cN_j}^{T_1}\right)\cdot M_{\VComb'}^{T_2}=\sum_{j=0}^1\ket{j}\ox I_{d^2}.
\end{equation}
Denote
\begin{equation}
    M_\VComb^{T_2}:=M_{\VComb'}^{T_2}\cdot(I_{d^2}-\ketbra{0}{0})+\ketbra{000}{0},
\end{equation}
and we claim $\VComb$ is a virtual comb satisfying \eqref{Eq:1slot_inverse_LS}.
Firstly, we could check \eqref{Eq:1slot_inverse_LS} as following:
\begin{align}
    \left(\sum_{j=0}^1\ket{j}\ox M_{\cN_j}\ox M_{\cN_j}^{T_1}\right)\cdot M_{\VComb}^{T_2}
    =&\left(\sum_{j=0}^1\ket{j}\ox M_{\cN_j}\ox M_{\cN_j}^{T_1}\right)\cdot \left(M_{\VComb'}^{T_2}\cdot(I_{d^2}-\ketbra{0}{0})+\ketbra{000}{0}\right)\\
    =&\sum_{j=0}^1\ket{j}\ox I_{d^2}\cdot(I_{d^2}-\ketbra{0}{0})+\left(\sum_{j=0}^1\ket{j}\ox M_{\cN_j}\ox M_{\cN_j}^{T_1}\right)\cdot\ketbra{000}{0}\\
    =&\sum_{j=0}^1\ket{j}\ox (I_{d^2}-\ketbra{0}{0})+\sum_{j=0}^1\ket{j}\ox M_{\cN_j}\ket{0}\ox M_{\cN_j}^{T_1}\ket{00}\bra{0}\\
    =&\sum_{j=0}^1\ket{j}\ox (I_{d^2}-\ketbra{0}{0})+\sum_{j=0}^1\ket{j}\ox \ket{0}\ox\bra{0}\\
    =&\sum_{j=0}^1\ket{j}\ox I_{d^2},
\end{align}
where $M_\cN\ket{0}=\ket{0}$, $M_\cN^{T_1}\ket{00}=1$ comes from $\cN_{j0}=\delta_{j0}$ in Lemma \ref{lemma:Pauli_TP}.
Since $M_\VComb^{T_2}\ket{0}=\ket{000}$, i.e. 
\begin{equation}
    \VComb_{klr0}=\delta_{k0}\delta_{l0}\delta_{r0},
\end{equation}
$\VComb$ always satisfies the condition for a virtual comb by Lemma \ref{lemma:Pauli_TP}.
As above all, $\VComb$ is already a virtual comb satisfying \eqref{Eq:1slot_inverse_LS} and we are done.
\end{proof}

\renewcommand\theproposition{\textcolor{blue}{3}}
\begin{theorem}
  For any two invertible quantum channels $\cN_1$, $\cN_2$, there exists a $1$-slot virtual comb $\VComb$ satisfying
  \begin{equation}
    \VComb(\cN_i) = \cN_i^{-1}, \quad \forall \, i=1,2.
  \end{equation}
\end{theorem}
\begin{proof}
The argument can be directly derived from Corollary \ref{Cor:comb_inverse_LS} and Corollary \ref{Cor:2_channel_inverse_comb}.
\end{proof}
\renewcommand{\theproposition}{S\arabic{proposition}}

\subsection{Reversing unitary channels}\label{appendix:unitary}

\renewcommand\theproposition{\textcolor{blue}{5}}
\setcounter{proposition}{\arabic{proposition}-1}
\begin{proposition}
For any dimension $d$, there exists a $1$-slot virtual comb $\VComb$ that transforms all qudit-unitary channels $\cU_d$ into their inverse $\cU_d^{-1}$, i.e., $\VComb(\cU_d)(\cdot) = U_d^{\dag}(\cdot)U_d$.
\end{proposition}
\begin{proof}
For simplicity, we label the systems $\cP,\cI,\cO,\cF$ as $1,2,3,4$, respectively.
Let
\begin{align}
    V_0 &= \frac{1}{d^2-1}\dketbra{I}{I}^{T_{3}}_{13}\ox\dketbra{I}{I}^{T_{4}}_{24} - \frac{1}{d(d^2-1)} I_{13} \ox \dketbra{I}{I}^{T_{4}}_{24} - \frac{1}{d(d^2-1)}\dketbra{I}{I}^{T_{3}}_{13} \ox I_{24} + \frac{1}{d^2-1}I_{1234}, \\
    V_1 &= -\frac{1}{d^2-1}\dketbra{I}{I}^{T_{3}}_{13}\ox\dketbra{I}{I}^{T_{4}}_{24} - \frac{1}{d(d^2-1)}I_{13}\ox \dketbra{I}{I}^{T_{4}}_{24} + \frac{1}{d(d^2-1)}\dketbra{I}{I}^{T_{3}}_{13} \ox I_{24} + \frac{1}{d^2-1} I_{1234}.\label{Eq:traceVO_0}
\end{align}
After simple calculation, we could find $V_0$ and $V_1$ are both positive, and it is straightforward to check that they are Choi operators of two quantum combs $\cC_0, \cC_1$, respectively. 
Then we have
\begin{align}
    V &= \frac{d^2}{2} V_0 - \frac{d^2-2}{2} V_1\\
        &=\dketbra{I}{I}^{T_{3}}_{13}\ox\dketbra{I}{I}^{T_{4}}_{24} - \frac{1}{d(d^2-1)}I_{13}\ox \dketbra{I}{I}^{T_{4}}_{24} - \frac{1}{d}\dketbra{I}{I}^{T_{3}}_{13} \ox I_{24} + \frac{1}{d^2-1} I_{1234}.
\end{align}
Denote $M_0 = \dketbra{I}{I}^{T_{3}}_{13}\ox\dketbra{I}{I}^{T_{4}}_{24}, M_1 = \frac{1}{d} I_{13}\ox \dketbra{I}{I}^{T_{4}}_{24}, M_2 = \frac{1}{d} \dketbra{I}{I}^{T_{3}}_{13} \ox I_{24}$ and $M_3 = I_{1234}/d^2$. Using the tensor network notations~\cite{wood2015tensor}, we have
\begin{equation}
M_0 = 
\diagram{
\draw (-6,1) .. controls (-4.5,1) and (-4.5,-0.8) .. (-2.5,-1);
\draw (-6,0) .. controls (-5,0) and (-4.5,-1.8) .. (-2.5,-2);
\draw (-6,-1) .. controls (-4.5,-1) and (-4.2,0.8) .. (-2.5,1);
\draw (-6,-2) .. controls (-4.5,-2) and (-3.5,-0.2) .. (-2.5,0);
},\quad
\dketbra{U}{U}^T = 
\diagram{
    \draw (-1,0) .. controls (-0.2,0) and (-0.2,-1) .. (-1,-1);
    \draw (1,0) .. controls (0.2,0) and (0.2,-1) .. (1,-1);
    
    \draw (-2.5,-1) -- (-1,-1);
    \draw (2.5,-1) -- (1,-1);
    \draw (-2.5,0) -- (-1,0);
    \draw (2.5,0) -- (1,0);
    \draw[ten] (-2,-0.5) rectangle (-1,0.5);
    \draw[ten] (2,-0.5) rectangle (1,0.5);
    \node at (-1.5,0) {$U^\dag$};	
    \node at (1.5,0) {$U$};
    },\quad
    \dketbra{U^\dag}{U^\dag} = 
    \diagram{
    \draw (-1,0) .. controls (-0.2,0) and (-0.2,-1) .. (-1,-1);
    \draw (1,0) .. controls (0.2,0) and (0.2,-1) .. (1,-1);
    
    \draw (-2.5,-1) -- (-1,-1);
    \draw (2.5,-1) -- (1,-1);
    \draw (-2.5,0) -- (-1,0);
    \draw (2.5,0) -- (1,0);
    \draw[ten] (-2,-1.5) rectangle (-1,-0.5);
    \draw[ten] (2,-1.5) rectangle (1,-0.5);
    \node at (-1.5,-1) {$U^\dag$};	
    \node at (1.5,-1) {$U$};
    }.
\end{equation}
Then it follows that
\begin{equation}
\diagram{
\draw (-6,1) .. controls (-4.5,1) and (-4.5,-0.8) .. (-2.5,-1);
\draw (-6,0) .. controls (-5,0) and (-4.5,-1.8) .. (-2.5,-2);
\draw (-6,-1) .. controls (-4.5,-1) and (-4.2,0.8) .. (-2.5,1);
\draw (-6,-2) .. controls (-4.5,-2) and (-3.5,-0.2) .. (-2.5,0);

\draw (-2.5,-2) -- (2.5,-2);

\draw (-6,0) .. controls (-7,0) and (-7,1.8) .. (-6,1.8);
\draw (2.5,0) .. controls (3.5,0) and (3.5,1.8) .. (2.5,1.8);
    \draw (-6,1.8) -- (2.5,1.8);

\draw (-6,-1) .. controls (-8,-0.5) and (-8,2.6) .. (-6,2.6);
\draw (2.5,-1) .. controls (4.5,-0.5) and (4.5,2.6) .. (2.5,2.6);
    \draw (-6,2.6) -- (2.5,2.6);

    \draw (-2.5,1) -- (2.5,1);
    \draw (-1,0) .. controls (-0.2,0) and (-0.2,-1) .. (-1,-1);
    \draw (1,0) .. controls (0.2,0) and (0.2,-1) .. (1,-1);
    \draw (-2.5,-1) -- (-1,-1);
    \draw (2.5,-1) -- (1,-1);
    \draw (-2.5,0) -- (-1,0);
    \draw (2.5,0) -- (1,0);

    \draw[ten] (-2,-0.5) rectangle (-1,0.5);
    \draw[ten] (2,-0.5) rectangle (1,0.5);

    \node at (-1.5,0) {$U^\dag$};
    \node at (1.5,0) {$U$};
    }
    = \diagram{
    \draw (-1,0) .. controls (-0.2,0) and (-0.2,-1) .. (-1,-1);
    \draw (1,0) .. controls (0.2,0) and (0.2,-1) .. (1,-1);
    
    \draw (-2.5,-1) -- (-1,-1);
    \draw (2.5,-1) -- (1,-1);
    \draw (-2.5,0) -- (-1,0);
    \draw (2.5,0) -- (1,0);
    \draw[ten] (-2,-1.5) rectangle (-1,-0.5);
    \draw[ten] (2,-1.5) rectangle (1,-0.5);
    \node at (-1.5,-1) {$U^\dag$};	
    \node at (1.5,-1) {$U$};
    },
\end{equation}
which gives $M_0 * \dketbra{U}{U} = \dketbra{U^\dag}{U^\dag}$. Similarly, we have
\begin{equation}
    M_1 * \dketbra{U}{U} = I_{\cP\cF}/d, M_2 * \dketbra{U}{U} = I_{\cP\cF}/d, M_3 * \dketbra{U}{U} = I_{\cP\cF}/d,
\end{equation}
which follows
\begin{equation}
    V * \dketbra{U}{U} = \dketbra{U^\dag}{U^\dag} - \frac{1}{d(d^2-1)} I_{\cP\cF} - \frac{1}{d} I_{\cP\cF} + \frac{d^2}{d(d^2-1)} I_{\cP\cF} = \dketbra{U^\dag}{U^\dag}.
\end{equation}
\end{proof}
\renewcommand{\theproposition}{S\arabic{proposition}}

\renewcommand\theproposition{\textcolor{blue}{6}}
\begin{theorem}
    The optimal sampling overhead for the $n$-slot virtual comb that can exactly reverse all $d$-dimensional unitary operations satisfies 
    \begin{equation}
        \nu(d,n) = \frac{2}{F_{\rm opt}(d,n)}-1,\quad \forall d\geq 2, n\geq 1,
    \end{equation}
    where $F_{\rm opt}(d,n)$ is the optimal average channel fidelity of reversing all $d$-dimensional unitary operations with an $n$-slot quantum comb.
\end{theorem}
\begin{proof}
    First, we are going to show 
    \begin{equation}\label{Eq:nu_geq_fid}
        \nu(d,n)\geq \frac{2}{F_{\rm opt}(d,n)} -1.
    \end{equation}
    Suppose $\{\widehat{\eta}, \widehat{\cC}_0, \widehat{\cC}_1\}$ is a feasible solution to the SDP \textcolor{blue}{(6)} of calculating $\nu(d,n)$ in the main text, where $\widehat{\cC}_0, \widehat{\cC}_1$ have Choi operators $\widehat{C}_0,\widehat{C}_1$, respectively. It follows
    \begin{equation}
        (1+\widehat{\eta})\tr[\widehat{C}_0\Omega] - \widehat{\eta}\tr[\widehat{C}_1\Omega] = 1,
    \end{equation}
    where $\Omega$ is the performance operator defined by 
    \begin{equation}
        \Omega:= \frac{1}{d^2}\int dU \dketbra{U^\dag}{U^\dag}_{\cP\cF}\ox \dketbra{U^*}{U^*}^{\ox n}_{\cI\cO}.
    \end{equation}
    Here $dU$ corresponds to the Haar measure, and $\dket{U} = (U\otimes I)\dket{I}$ corresponds to the Choi operator of unitary $\cU$.
    Notice that $\widehat{\eta} \geq 0$ and $\tr[\widehat{C}_1\Omega]\geq 0$ will yield $\tr[\widehat{C}_0\Omega]\geq 1/(1+\widehat{\eta})$. Since $\widehat{\cC}_0$ is a quantum comb, it is a feasible solution to the SDP for maximum fidelity which gives $F_{\rm opt}(d,n) \geq 1/(1+\widehat{\eta})$. Therefore, we have $2\widehat{\eta} + 1\geq 2/F_{\rm opt}(d,n)-1$ and $\nu(d,n)\geq 2/F_{\rm opt}(d,n)-1$. 
    Second, we shall prove that 
    \begin{equation}\label{Eq:nu_leq_fid}
        \nu(d,n)\leq \frac{2}{F_{\rm opt}(d,n)} -1.
    \end{equation}
    Before that, we demonstrate that there always exists a quantum comb $\cC'$ with a Choi operator $C'$ such that $\tr[C'\Omega] = 0$ as follows. Consider a Choi operator $C' = V * W_1 * W_2*\cdots *W_{n-1}$ where $V$ is as defined in Eq.~\eqref{Eq:traceVO_0}
    \begin{equation}
    \begin{aligned}
    V = &-\frac{1}{d^2-1}\dketbra{I}{I}^{T_{\cO_1}}_{\cP\cO_1}\ox\dketbra{I}{I}^{T_{\cF_1}}_{\cI_1\cF_1} - \frac{1}{d^2-1}I_{\cP\cO_1}\ox \dketbra{I}{I}^{T_{\cF_1}}_{\cI_1\cF_1}\\ 
    &+ \frac{1}{d^2-1}\dketbra{I}{I}^{T_{\cO_1}}_{\cP\cO_1} \ox I_{\cI_1\cF_1} + \frac{1}{d^2-1} I_{\cP\cI_1\cO_1\cF_1},
    \end{aligned}
    \end{equation}
    and $W_k = I_{\cF_k\cO_{k+1}}\ox \Phi_{\cI_{k+1}\cF_{k+1}}, \cF_{n} := \cF$. We note that $V$ is semidefinite positive and satisfies the linear constraints of a quantum comb, and $W_k$ corresponds to a 1-slot comb which outputs an identity channel for any input channel. It is straightforward to check that for any input unitary channel $\cU_d$
    \begin{equation}
    \begin{aligned}
        \tr \left[\dketbra{U^\dag}{U^\dag} \cdot (C'* \dketbra{U}{U}^{\ox n})\right] &= \tr\Big[\dketbra{U^\dag}{U^\dag} \Big(V*\dketbra{U}{U}*\overbrace{\dketbra{I}{I}*\cdots *\dketbra{I}{I}}^{n-1}\Big)\Big]\\
        &= \tr\Big[\dketbra{U^\dag}{U^\dag} (V*\dketbra{U}{U})\Big]\\
        &= 0, \quad\forall U\in SU(d),
    \end{aligned}
    \end{equation}
    which indicates $\tr[C'\Omega] = 0$.
    Now suppose comb $\widehat{\cC}$ is a feasible solution to the SDP for maximum fidelity with a Choi operator $\widehat{C}$. We denote $f = \tr[\widehat{C}\Omega]$ and construct a virtual comb $\widetilde{\cC}'$ with a Choi operator 
    \begin{equation}
        \widetilde{C}' = \frac{1}{f}\widehat{C} - \frac{1-f}{f}C'
    \end{equation}
    We can see $\tr[\widetilde{C}'\Omega]=1$ and $\widetilde{\cC}'$ is a feasible solution to the SDP of calculating $\nu(d,n)$ in the main text. Therefore, we have $\nu(d,n)\leq 1/f + (1-f)/f = 2/f - 1$ and $\nu(d,n)\leq  2/F_{\rm opt}(d,n) - 1$. Combining Eq.~\eqref{Eq:nu_geq_fid} and Eq.~\eqref{Eq:nu_leq_fid}, we conclude that
    \begin{equation}
        \nu(d,n)= \frac{2}{F_{\rm opt}(d,n)} -1.
    \end{equation}
\end{proof}
\renewcommand{\theproposition}{S\arabic{proposition}}

Based on the result for the optimal channel fidelity shown in Ref.~\cite{Quintino2022}, we provide the optimal sampling overhead for different $d$ and $n$ in Table~\ref{tab:uinv_opt_cost}. 
An interesting case occurs when $d=2$, $n=3$ where we could find that $n\cdot\nu^2(2,3) \leq 3.9235$, which is strictly smaller than the $4$ required by the deterministic and exact method by a conventional quantum comb.
\begin{table}[H]
    \centering
    \tabcolsep=0.7cm
    \renewcommand\arraystretch{1.5}
    \footnotesize
    \begin{tabular}{cccccc}
    \specialrule{.12em}{0pt}{.1ex}
    & $n=1$ & $n=2$ & $n=3$ & $n=4$ & $n=5$ \\
    \hline
    $d=2$ & 3.0000 & 1.6667 & 1.1436 & 1.0000 & 1.0000 \\
    \hline 
    $d=3$ & 8.0000 & 5.0000 & 3.5000 & 2.6000 & 2.0000 \\ 
    \hline 
    $d=4$ & 15.0000 & 9.6667 & 7.0000 & 5.4000 & 4.3691 \\ 
    \hline 
    $d=5$ & 24.0000 & 15.6667 & 11.5000 & 9.0000 & 7.3333 \\
    \hline 
    $d=6$ & 35.0000 & 23.0000 & 17.0000 & 13.3988 & 11.0000\\ 
    \bottomrule[1pt]
    \end{tabular}
    \caption{The optimal sampling overhead for an $n$-slot virtual comb that can exactly reverse arbitrary $d$-dimensional unitary operations.}
    \label{tab:uinv_opt_cost}
\end{table}

\begin{figure}
    \centering
    \includegraphics[width=.65\linewidth]{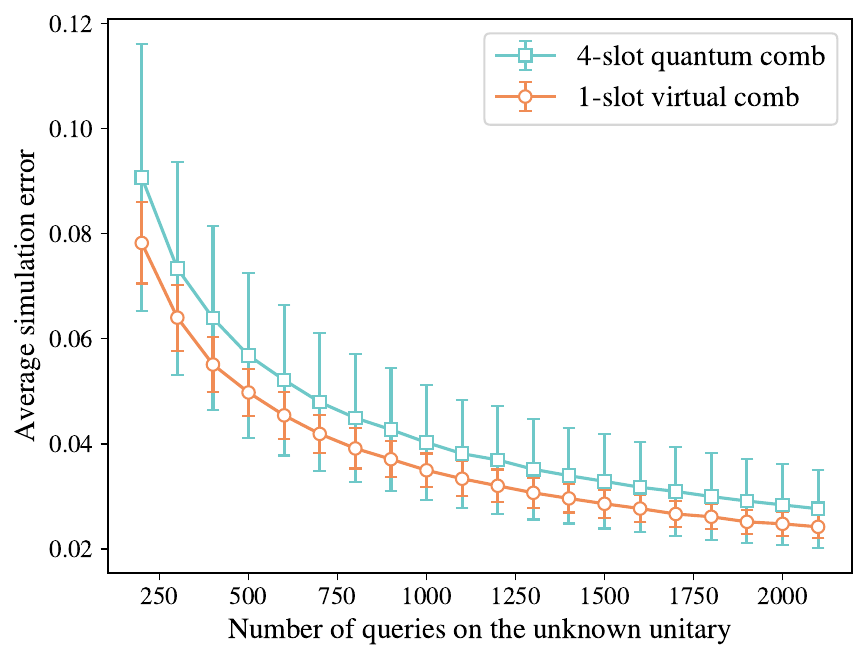}
    \caption{Average error for 200 randomly generated single-qubit unitary operations across various numbers of unitary queries. The data points represent the average absolute error between the empirical mean value and the theoretical expectation value, $\tr[ZU^\dag \ketbra{0}{0} U]$, for a given number of queries to the unknown unitary. The $x$-axis corresponds to different numbers of queries. The protocols using a 1-slot virtual comb query the unitary a single time, whereas the protocols using a 4-slot quantum comb conduct four queries, in a single measurement shot.}
    \label{fig:mea_simulation}
\end{figure}

Although the sufficient querying number of the unitary for $n=1$ is determined by $1\cdot\nu^2(d,1)$, scaling as $\cO(d^4)$, which is worse than the $\cO(d^2)$ required by the deterministic and exact protocol~\cite{chen2024quantum}, here we show that when the state $\rho$ or the observable $O$ is given, the query complexity could be significantly reduced. 
A specific example is shown in Fig.~\ref{fig:mea_simulation} for qubit wherein the system is initially prepared in the state $\ketbra{0}{0}$ and the observable $O$ is set to be the qubit Pauli $Z$ operator. 
In this setup, we show that the 1-slot virtual protocol has a better performance in both average simulation error and standard deviation under the same number of queries of the unknown unitary gate in estimating the quantity $\tr[ZU^{\dag}\ketbra{0}{0} U]$, compared with the deterministic and exact protocol~\cite{Yoshida2023,chen2024quantum}.

The previous method in Ref.~\cite{Yoshida2023} employs a 4-slot quantum comb to deterministically and exactly reverse the unknown $U$, where the quantity $\tr[OU^{\dag}\ketbra{0}{0} U]$ could then be obtained by inserting a qubit in $\ket{0}$ and measuring the output state corresponding to observable $O$. 
In our framework, we shall find a virtual comb $\VComb$ such that $\VComb(\cU)(\ketbra{0}{0}) = U^{\dag}\ketbra{0}{0} U$ for any unitary operation $U$. We numerically found the virtual comb satisfies the above with an optimal sampling overhead equal to $1.5$. Using this comb, we did the following numerical simulation. 
\begin{enumerate}
    \item Sample $200$ random qubit unitary operations according to the Haar measure.
    \item Choose different numbers of measurement shots to make sure the two methods have the same number of queries for the unknown unitary. We apply numerical simulation to obtain the absolute errors between these two methods' empirical mean and the theoretical expectation value.
    \item Repeat the simulation $200$ times for each unitary and output the mean absolute errors as simulation errors of this unitary.
    \item For each number of queries, repeat Step 2 and Step 3 and compute the average simulation errors for $200$ randomly generated unitary operations. 
\end{enumerate}
In Fig.~\ref{fig:mea_simulation}, we plot the average simulation errors over $200$ randomly generated qubit unitary operations, with varying numbers of queries. The $x$-axis corresponds to different query numbers and the $y$-axis corresponds to the average simulation errors. We observe that when the number of queries is identical, the average simulation error is smaller by utilizing a 1-slot virtual comb compared with the previous method using a 4-slot quantum comb. To further evaluate the efficacy of virtual combs relative to quantum combs, we have searched the optimal quantum combs capable of executing this task with fixed input state $\ketbra{0}{0}$. By solving SDPs, we found that there is a $3$-slot quantum comb that can successfully perform the task, while there is no solution with a $2$-slot quantum comb. Note that a requisite number of queries to attain a certain measurement precision is governed by $n\cdot \nu^2(d,n)$ where $d$ is the dimension of the system and $n$ is the slot number. Our findings reveal that the 1-slot virtual comb yields $1 \cdot 1.5^2 = 2.25$, which is notably less than $3 \cdot 1^2 = 3$ required by a 3-slot quantum comb.

Therefore, although a $1$-slot virtual comb has a worse querying count scaling for general $d$ compared with methods using quantum combs, it practically could perform better, e.g., in a qubit case with a fixed input state. This enhancement in accuracy underscores the potential of virtual combs in reducing the querying complexity for estimating $\tr[OU^{\dag}\rho U]$ by querying $U$ when the input state or the observable information is known. This also gives insights into shadow tomography, suggesting that the manipulation of the process and the sampling procedure could be jointly optimized to streamline query complexity.

\section{The protocol for error cancellation}\label{appendix:protocol}
In this part, we give the detailed protocol for achieving high-precision error cancellation for depolarizing channels with unknown parameters within a given range.
In particular, if we aim to cancel the effect of an unknown depolarizing noise $\cD_{p}$ from a set of distinct noise parameters simultaneously as described in Theorem \textcolor{blue}{1}, we have the following protocol to obtain an estimation of $\tr[O\rho]$ instead of $\tr[O\cD_{p}(\rho)]$. We denote three types of operations for our protocol:
\begin{enumerate}
    \item[1)] $\cC_{\id}$: Do nothing to the received state.
    \item[2)] $\cC_{\cD}$: Replace the received state with a maximally mixed state.
    \item[3)] $\cC_i$: Apply the unknown process to the received state iteratively for $i$ times.
\end{enumerate}
Based on the above, the protocol is shown as Protocol~\ref{alg:protocol_depo}.
\begin{algorithm}\label{alg:protocol_depo}
    \renewcommand{\algorithmcfname}{Protocol}
    \SetKwInOut{Input}{Input}\SetKwInOut{Output}{Output}
    \caption{Error cancellation for getting $\tr[O\rho]$ through unknown depolarizing noise\label{algo:inv_depo}}
    \Input{Unknown depolarizing process $\cD_{p}$ with $p$ in a given set $\{p_1, ..., p_{n+1}\}$, unknown state after depolarizing noise $\cD_{p}(\rho)$, a given observable $O$, and given error parameters $\epsilon, \delta \in (0,1)$;
    }
    \Output{An estimator $\zeta$ of $\tr[O \rho]$}
    Divide the inversion virtual comb into different implementable processes: $\VComb = \eta_{id} \cdot \cC_{id} + \eta_{\cD}\cdot \cC_{\cD} + \sum_{1}^{n} \eta_{i} \cdot \cC_i$, with coefficients given by the solution of the linear equations in Eq.~\eqref{appendEq:Nslot_inverse} determined by $\{p_1, ..., p_{n+1}\}$ and $n$\;
    Calculate the sampling rounds for error cancellation: $S = \lceil 2\gamma^2\cdot \log(2/\delta)/\epsilon^2 \rceil$, where $\gamma = |\eta_{id}| + |\eta_{\cD}| + \sum_i |\eta_i|$\;
    \For{$s$ from $1$ to $S$}{
        Sample an operation $\cC_s$ from $\{\cC_{\id},\cC_{\cD},\cC_1,\cC_2,...,\cC_n\}$ with probability $|\eta_s|/\gamma$\;
        For the input state $\cD_{p}(\rho)$, which has passed through the unknown depolarizing process, apply the randomly sampled inversion process $\cC_s$, followed by a measurement of the corresponding observable $O$, with the result denoted as $\lambda_s$\;
        $X^{(s)} \leftarrow \gamma\cdot{\rm sgn}(\eta_{s}) \lambda_s$ \;
    }
    Get the estimator $\zeta$ of $\tr[O \rho]$: 
    $\zeta\leftarrow \frac{1}{S} \sum_{s=1}^S X^{(s)}$\;
    Return $\zeta$\;
\end{algorithm}

It is encouraging that the three types of operations considered in this protocol are quite simple. Notably, if the depolarizing noise is from a continuous region, i.e., $p\in[p_1,p_{n+1}]$, from the analysis presented in Theorem \textcolor{blue}{4}, we could notice that with an increasing number of slots in the virtual comb, Protocol~\ref{alg:protocol_depo} exhibits a worst-case error diminishing at least as $\cO(n^{-1})$. The precision gains hold valid for a broad class of depolarizing channels, ensuring robustness across a continuum of noise levels. It is foreseeable that, as this approach allows for noise within a certain range while maintaining high-precision performance, it will reduce the initial precision requirements for error tomography. It will also ensure the robustness in overall operation, i.e. the protocol remains effective even if the error changes a little bit during the experiment.

The optimal sampling overhead can be calculated as the following optimization problem.
\begin{subequations}\label{SDP:cost}
\begin{align}
\gamma_{\rm opt}= \min &\;\; 2\eta + 1,\\
 {\rm s.t.} & \;\; \VComb = (1+\eta) \cC_0 - \eta \cC_1, \\
            & \;\; \cC_0, \cC_1 \text{ are $n$-slot quantum combs,}\label{Eq:chan_comb_cons}\\
            & \;\; \VComb(\cN_i)\circ\cN_i = \id \,\,\, \forall i=0,1,...,n \,.      
\end{align}
\end{subequations}
Notice that $\eta$ and $\cC_0,\cC_1$ are all variables. We can absorb $\eta$ into $\cC_i$ and make all constraints linear by rewriting $\VComb = \cC_0 - \cC_1$ and modify the trace constraints in Eq.~\eqref{Eq:chan_comb_cons} by $\tr C_0 = (1+\eta)d^2$ and $\tr C_1 = \eta d^2$ where $C_0, C_1$ are the Choi operators of $\cC_0, \cC_1$, respectively. Thus, the optimization problem is an SDP.

\section{Connections to probabilistic error cancellation and shadow tomography}\label{appendix:connection}
We have addressed the reversibility of quantum processes using a novel framework of virtual combs which is a foundational problem in quantum mechanics. More importantly, the virtual comb framework highlights the power of combining higher-order quantum processes with classical computation, especially through sampling and post-processing. This approach has broader implications for many other quantum methodologies, including quantum error mitigation, shadow tomography, randomized quantum algorithms, etc. We provide further discussion in this section.

\paragraph{Connection to probabilistic error cancellation:} The concept of virtual combs extends the traditional approach of simulating unphysical Hermitian-preserving trace-nonincreasing maps, which serve as the foundational technique in \textit{probabilistic error cancellation}~\cite{Temme_2017}. Indeed, a probabilistic error cancellation protocol can be viewed as a particular instance of a virtual comb, wherein the pre-processing channel of the virtual comb is simply the identity operation~\cite{Temme_2017,Takagi2021,Jiang2020}. Furthermore, by adopting a higher-order operational framework, virtual combs unveil a collection of novel and insightful properties. For example, a uniquely determined virtual comb can reverse different channels, which is impossible for a traditional probabilistic error cancellation. Another thing we would like to mention is that through numerical calculation, we find that there is a 2-slot virtual comb that can simulate the inverse of a unitary $U$ by querying $\cU \circ \cD_{p}$, where $U$ is affected by some depolarizing noise and both $U \in SU(d)$ and $p\in \{p_1, p_2\}$ are unknown. 
It is an interesting result and may inspire studies about unitary inversion with noise.
As the virtual combs enable the most comprehensive manipulation of quantum noise, they offer a promising avenue for enhancing probabilistic error cancellation methodologies.

\paragraph{Connection to shadow tomography:} The practical simulation of virtual combs, which entails estimating expectation values, is closely associated with \textit{shadow tomography}~\cite{aaronson2018shadow}. In detail, the allowance for negative coefficients within a virtual comb may pave the way for a more universal protocol description that encompasses shadow tomography tasks, or other general operations processing. The virtual comb framework may also inspire other strategies to reduce the querying complexity as the fundamental insight of a virtual comb lies in its trade-off: a virtual comb using fewer slots sacrifices the quantum memory effect inherent to multi-slot quantum combs in favor of sampling and classical post-processing, to achieve an equivalent level of information processing capability. However, a more thorough analysis and investigation into the sampling costs and specific tasks are necessary for substantiating these claims, which we anticipate will be a fruitful direction for future research.

\paragraph{Connection to quantum algorithm design:} The virtual comb framework could inspire new designs for quantum algorithms, especially those that require inversions of operations, such as the quantum singular value transformation (QSVT)~\cite{Gilyen2018}. It helps us to understand better the capabilities of performing higher-order transformations.
The idea of using virtual combs for process inversion can also be related to the concept of classical shadows~\cite{Huang_2020}, which is a technique where classical data obtained from randomized measurements is used to predict many properties of a quantum state. Virtual combs could extend this idea to the domain of quantum operations, potentially leading to new randomized algorithms for estimating the effects of quantum channels.

\end{document}